\documentclass[11pt]{amsart}

\usepackage[hmargin=1.0in,height=8.3in]{geometry}
\usepackage{amsthm,amsmath,epsfig,subfigure,amsfonts,amssymb,bm,xcolor,amsbsy,times}
\usepackage{bm}
\usepackage{bbm}
\usepackage{enumerate}
\usepackage{mathrsfs}
\usepackage{graphicx}
\usepackage{comment}
\usepackage{mathtools}
\usepackage{cases}
\usepackage{dsfont, amssymb,amsmath,amscd,latexsym, amsxtra,amsfonts,amsthm}
\usepackage[round]{natbib}
\usepackage[pdfstartview=FitH, bookmarksnumbered=true,bookmarksopen=true, colorlinks=true, pdfborder=001, citecolor=blue, linkcolor=blue,urlcolor=blue]{hyperref}
\numberwithin{equation}{section}
\graphicspath{{Numerical/}}
\usepackage{ifpdf}
\usepackage{color}
\definecolor{webgreen}{rgb}{0,.5,0}
\definecolor{webbrown}{rgb}{.8,0,0}
\definecolor{emphcolor}{rgb}{0.95,0.95,0.95}

\newtheorem{theorem}{Theorem}[section]
\newtheorem{corollary}[theorem]{Corollary}
\newtheorem{proposition}[theorem]{Proposition}
\newtheorem{lemma}[theorem]{Lemma}
\theoremstyle{definition}
\newtheorem{definition}{Definition}[section]

\theoremstyle{definition}

\theoremstyle{definition}
\newtheorem{remark}{Remark}[section]
\theoremstyle{definition}
\newtheorem{assumption}{Assumption}[section]

\renewcommand{\epsilon}{\varepsilon}

\makeatletter
\newcommand*\bigcdot{\mathpalette\bigcdot@{.5}}
\newcommand*\bigcdot@[2]{\mathbin{\vcenter{\hbox{\scalebox{#2}{$\m@th#1\bullet$}}}}}
\makeatother
\numberwithin{equation}{section}

\title[  ]{The continuous-time pre-commitment KMM problem in incomplete markets}
\author[  ]{Guohui Guan}
\address[G. Guan]{School of Statistics, Renmin University of China, Beijing 100872, China.}
\email{guangh@ruc.edu.cn}
\author[  ]{Zongxia Liang}
\address[Z. Liang]{Department of Mathematical Sciences, Tsinghua University, Beijing, China.}
\email{liangzongxia@mail.tsinghua.edu.cn}
\author[  ]{Yilun Song}
\address[Y. Song]{Department of Mathematical Sciences, Tsinghua University, Beijing, China.}
\email{songyl18@mails.tsinghua.edu.cn}

\begin{document}
\begin{abstract}
 This paper studies the continuous-time pre-commitment KMM problem  proposed by  
\cite{7} {{in incomplete financial markets}}, which concerns with the portfolio selection under smooth ambiguity. The decision maker (DM) is uncertain about the {{dominated}} 
 priors of the financial market, which are characterized by a second-order distribution (SOD). The KMM model separates risk attitudes and ambiguity attitudes apart and the aim of the DM is to maximize the two-fold utility of terminal wealth, which does not belong to the {{classical}} subjective utility maximization problem. By constructing the efficient frontier, the original KMM problem is first simplified as {{an one-fold expected utility problem on the second-order space.}} 
In order to solve the equivalent simplified problem, this paper {{imposes an assumption and}} introduces a new distorted Legendre transformation to {{establish the bipolar relation and the distorted duality theorem. 
Then, under a further assumption that the asymptotic elasticity of the ambiguous attitude is less than 1, 
the uniqueness and existence of the solution to the KMM problem are shown }}
and we obtain the semi-explicit forms of the optimal terminal wealth and the optimal strategy. 
Explicit forms of optimal strategies are presented for CRRA, CARA and HARA utilities in the case of Gaussian SOD in a Black-Scholes financial market, which show that DM with higher ambiguity aversion tends to be more concerned about extreme market conditions with larger bias. In the end of this work, numerical comparisons with the DMs ignoring ambiguity are revealed to illustrate the effects of ambiguity on the optimal strategies and value functions.
\end{abstract}
\maketitle
\section{\bf{Introduction}}
In this paper, we are interested in continuous-time portfolio selection problem under smooth ambiguity {{in incomplete financial markets}}. Because the decision maker (DM) may be worried by  cognitive or informational constraints which make him/her uncertain about the financial market, the term ``ambiguity'' is proposed in \cite{E61}. When making decisions, the DMs are uncertain about the {{precise}} 
distributions {{(the real-world probability measure)}} of the assets  and often adopt subjective estimations about the financial market. Therefore, ambiguity should be considered. 
\cite{7} present the smooth ambiguity model (abbr. KMM) showing that the DM  is (subjectively) uncertain about the priors relevant to his/her decision. The beliefs of DM are modelled by a set of {{dominated}} probability measures 
{{and}} a second-order distribution (SOD). More precisely, we mainly investigate the following optimization problem under smooth ambiguity:
 \begin{equation}\label{U}
\max_\pi\left\{ \int_{\mathbb{D}}\phi\left(\mathbf{E}^{Q^\mu}\left[U({{X^\pi_T}})\right]\right)\mathrm{d}F(\mu)\right\},
\end{equation}
where ${{X^\pi_T}}$ is the contingent claim at time $T$ associated with portfolio $\pi$. $Q^\mu$ represents the DM's subjective estimation 
and reflects the DM's certain prior about the financial market. The uncertainty over the priors is modelled by the SOD $F(\mu)$, which reflects the DM's confidence  of prior $Q^\mu$ {{and represents the DM's belief about the financial market}}. $U(\cdot)$ and $\phi(\cdot)$  are two functions characterizing the DM's risk aversion and ambiguity aversion, respectively. The above KMM model has been applied to explain the Ellsberg's paradox and can separate attitude towards risk from that towards ambiguity, which makes it  popular in economics.
\vskip 5pt

In the original KMM problem of 
\cite{7}, there {{is a non-linear weighted average of expected utilities}} 
leading to time-inconsistency, which makes the problem intractable by stochastic dynamic programming method. \cite{BJ17} presents a general time-inconsistent optimization criterion. However, the form in \cite{BJ17} only contains one-fold expectation. On the one hand, Eq.~(\ref{U}) has a non-linear function of the expected utility, which causes time-inconsistency as in \cite{BJ17}. On the other hand, {{when the non-linear weighted average is treated as another expectation of a utility, the objective becomes the two-fold expectations of utilities, which also}} brings in time-inconsistency and makes the problem more complicated. As for the discrete-time recursive smooth ambiguity model in \cite{Klibanoff}, the problem  becomes time-consistent, but explicit  or semi-explicit solutions can be hardly derived.  As such, only numerical optimization methods are applied to approximate its solution in most work. The problems of maximizing the utility under smooth ambiguity of \cite{Klibanoff} have been extensively studied (cf. \cite{Ju12}, \cite{Chen14}, \cite{Co16}), as such, in the presence of smooth ambiguity, the equity premium puzzle can be potentially explained to some extent. Recently, \cite{Balter et21} consider a continuous-time KMM model. However, they require the strategy to be deterministic and the settings in their model are restricted to the Black-Scholes model.
\vskip 5pt
It is well known that utility maximization is a popular topic for DMs in mathematical finance. The classical maximization problem of von Neumann-Morgenstern expected utility has been studied under different backgrounds. Under suitable conditions, the duality theorem and existence of solution are first revealed in \cite{2} for a complete market and 
\cite{25}, \cite{Hu05} for an incomplete market. As the dual problem is often simpler than the primal one, duality is an efficient tool for more general financial models. We refer the readers to \cite{C01}, \cite{Hu04}, \cite{Lin17}, \cite{Wit} with  random endowment, \cite{Bia05}, \cite{Owen09} with unbounded processes, \cite{Li18} with constraints, etc.
\vskip 5pt
However, most previous work neglect {{model}} uncertainty, in which the DM obtains {{the}} point estimation 
{{of the real-world measure of the financial market and often assumes that the estimation is precise and reliable,}} 
i.e., the DM knows exactly the true distribution (first-order distribution, FOD) of the financial assets. However, {{for instance, even some of}} the parameters are hard to be estimated accurately, particularly for the first moments of the yields, see \cite{Blanchard93} and \cite{BT18}. As such, the FOD is in fact imprecise and the DMs should take the ambiguity over FOD into account. Concerning the worst case scenario proposed in \cite{Gil89}, the duality in the max-min concept is revealed in \cite{Sch05}. \cite{Sch05} proves the existence of the worst case measure and transforms the robust problem into an equivalent standard utility maximization problem w.r.t. this measure. Duality theorem is easily applied for the transformed standard problem. \cite{Sch07} further develops the duality theory for the maximization of the robust utility {{for dominated priors}} in a very general setting and under rather weak assumptions. {{Then \cite{Bartl21} establish a duality theory for the robust utility maximisation problem in continuous time for possibly non-dominated priors.}} 
More related work about max-min concept can refer to \cite{Gun05},  \cite{Sch08}, \cite{Tev13},  \cite{Back16}, \cite{Neu}, etc.
\vskip 5pt

Although there are many discussions about the  expected utility and max-min concept, the existence {{of the solution}} and {{the duality theorem}} have not yet been studied in the KMM model. On the one hand, most of the previous work are  concerned with one-fold expectation and utility. Then, by dual method, the Legendre transformation of the one-fold problem can be well established. The form of Eq.~(\ref{U}) shows that it contains two expectations and two  utilities, in which {{it is difficult to define even the conjugate function and}} the dual method cannot be applied directly. To the best of our knowledge, this paper is also the first to investigate {{the duality theorem of}} the KMM problem in an incomplete market in a continuous-time framework. There are little literature about the original KMM model within continuous-time market. On the other hand, the KMM problem is time-inconsistent. In \cite{BJ17}, the time-inconsistent feature is revealed for the optimization goal with the non-linear function of the expectation of the utility. We see from Eq.~(\ref{U}) that the KMM model extends the time-inconsistent optimization form in \cite{BJ17}. Recently, the time-inconsistent problems have been widely investigated, see \cite{Ekeland et12} with non-exponential discount factor, \cite{BJ13} for the mean-variance criterion, \cite{Hu et17} for linear-quadratic control, etc.  The DM fixes a given target at the initial time and keeps it unchanged over time in the pre-commitment case. In \cite{Vigna20}, the author shows that  the pre-commitment strategy beats the Nash equilibrium strategy and dynamically optimal strategy in the mean-variance criterion. Besides, the uniqueness of the Nash equilibrium strategy is not revealed in the game approach.  Different from the Nash equilibrium strategy in \cite{BJ17}, we concentrate on the original pre-commitment KMM model in 
\cite{7}. As in \cite{Vigna20}, the pre-commitment strategy can be easily extended to the dynamically optimal strategy. In this paper,  we aim to derive the {{semi-}}explicit solution within continuous-time market. For simplicity, we assume that the financial market includes a risk-free asset and a risky asset. The price of the risky asset is assumed to be a semi-martingale {{under}} 
different priors, {{which are dominated by some observable measure. The bipolar relation and the duality theorem}} 
are shown for the general financial model. Besides, the financial market with Black-Scholes model is studied as a specific example {{and we obtain explicit solution in the example}}.
\vskip 5pt

Because the pre-commitment KMM problem is totally different from the one-fold expected utility problem and max-min problem, the duality in 
\cite{25} for a standard utility or \cite{Sch07} for robust control cannot be applied here. The duality can often only be applied for one-fold {{expected}} utility maximization. For the robust control, {{both \cite{Sch07} and \cite{Bartl21} derive the duality theorem on the first-order space, i.e., on the sets of super-hedgeable claims and probability measures. 
In order to derive the duality theorem}} 
of the two-fold {{objective}} 
function of  Eq.~(\ref{U}), {{we aim to separate the problem into two problems with one-fold expected utilities and establish the duality theorem for each problem. As such, for the problem on the first-order space,}} we first characterize the efficient frontier (the set of the efficient strategies which are not dominated by others) for different priors of the admissible set by convex analysis, where the efficient ones must maximize
a {{linear}} weighted average of optimal expected utilities for an arbitrary weight function. With this efficient frontier,
we simplify the original optimization with the two-fold expectations  as the maximization of the one-fold expectation among all efficient terminal claims, {{which are considered on the second-order space, i.e., on the sets of expected utilities and weight functions}}. 
Afterward, 
{{under suitable assumptions  about the frontier, we establish the bipolar relation and the distorted duality theorem for the problem on the second-order space. The semi-explicit form of the optimal strategy is also obtained.}}
\vskip 5pt

Next, we study some concrete examples for specific $U(\cdot)$ and $ \phi(\cdot)$ and obtain the optimal investment strategy explicitly in a Black-Scholes financial market with SOD characterized by a normal distribution. Besides, we also find that this problem with the value $xe^{rT}$ of the budget constraint is equivalent to the optimization problem $\int_{\mathbb{D}}\phi(\mu)\mathrm{d}F(\mu)$ on the ambiguous range $\mathbb{D}$ with the value $U(xe^{rT})$ of the distorted budget constraint in these specific cases. Different from the case without ambiguity, the optimal terminal wealth relies on a quadratic function of the {{observable}} market state. Meanwhile, the strategy also depends on the {{observable}} market states.  DM with ambiguity aversion would rather give up some benefits of normal situations with small bias to ensure the benefits of extreme situations with large bias.
Finally, we provide some comparisons with the DM's ignoring ambiguity and show sensitivity analyses to illustrate the effects of risk attitudes and ambiguity attitudes on the optimal investment strategies and value functions.
\vskip 5pt

We have the following contributions in this paper. First, we formulate the pre-commitment KMM model {{for dominated priors in an incomplete financial market}} in a continuous-time framework. Although the discrete-time KMM model in \cite{Klibanoff} has been widely applied, the original continuous-time KMM model of 
\cite{7} has not yet been studied.  Second, as the KMM problem is a combination of two-fold expectations and two-fold utilities, classical optimization method cannot be applied. {{By convex analysis,}} we transform the original {{two-fold expected utility optimization}} problem {{equivalently to the combination of two kinds of one-fold expected utility optimization problems (see Problem~(\ref{pr2}) on the first-order space and Problem~(\ref{p3}) on the second-order space). Problem~(\ref{pr2}) on the first-order space is just a classical EUT problem in an incomplete financial market, which ensures that all of the results in \cite{25} holds for Problem~(\ref{pr2}) on the first-order space. Then we concentrate on Problem~(\ref{p3}) on the second-order space. Third, under Assumption \ref{Ass}, we introduce the distorted Legendre transformation and derive the bipolar relation in Theorem \ref{bipolarThm} and the distorted duality theorem in Theorem \ref{duality} for Problem~(\ref{p3}). In addition, if the asymptotic elasticity of the ambiguous attitude is less than 1, the existence and uniqueness of the solution to Problem~(\ref{p3}) and the dual relation of the solutions to Problem~(\ref{p3}) and the dual problem (\ref{D2}) are obtained in Theorem \ref{solution}.
}} 
At last, we present the optimal wealth and strategy under three different utility functions in the case that the SOD is Gaussian distributed. The results show that the optimal terminal wealth relies on  a quadratic function of market states which is quite different from the linear relationship when ignoring ambiguity, which leads to the fact that the ambiguity averse DM concentrates more on extreme situations than the ambiguity neutral DM. {{
Moreover, the optimal investment proportion is a constant at time $t=0$ and relies on the observable market states for $t>0$.}} In order to derive the explicit solution, we only consider the uncertainty of drift term in a Black-Scholes financial market. However, the formal results of the efficient frontier, {{distorted}} duality theorem can also be established for the uncertainty of volatility, etc.

\vskip 5pt
The rest of this paper is organized as follows. The financial market, DM's wealth process and smooth ambiguity {{model}} are presented in Section 2. In Section 3, {{properties of the efficient frontier, the bipolar relation and the duality theorem on the second-order space are  
derived {in Theorems \ref{TT1}, \ref{bipolarThm}, \ref{duality}, and the semi-explicit solution of the KMM problem is given in Theorem \ref{TT5}}.}} 
Section 4 presents the optimal investment strategy {explicitly} for different  $U(\cdot)$ and $ \phi(\cdot)$ in a Black-Scholes financial market {based on Theorem \ref{TT5}}. Section 5 shows the numerical results and Section 6 is a conclusion.
\setcounter{equation}{0}
\section{\bf Problem Formulation}
\label{Problem formulation}
In this section, we present the financial model under model uncertainty. The optimization problem under smooth ambiguity is also shown.
\subsection{\bf Financial market}

We consider an {{incomplete}} financial market 
in a complete filtered  probability space 
$(\Omega, \mathcal{F}, \{\mathcal{F}_t\}_{0\leq t\leq T}, \mathbf{P})$, where $\mathcal{F}_t$ is the information of financial market up to time $t$ and $[0,T]$ is a fixed time horizon.  $\mathbf{P}$ is the real probability measure of the financial market. In what follows, all the processes introduced below are assumed to be well defined and adapted to $\{\mathcal{F}_t\}_{0\leq t\leq T}$. We do not consider transaction costs and short selling is allowed.
In the financial market,  
we assume that there are two assets: a risk-free asset {{$S^0$}} and a risky asset {{$S^1$}}. The risk-free asset {{$S^0$}} satisfies
{{
\[\mathrm{d}S^0_t=rS^0_t\mathrm{d}t,~S^0_0=1,\]
}}
where $r>0$ is the risk-free rate. {{The natural filtration generated by $S^1$ is denoted by $\{\mathcal{F}^S_t\}_{0\leq t\leq T}$.}}  
\vskip 3pt
As the DM may be worried by the limitation of information which makes him/her uncertain about the precise distributions of the financial market, we assume that the real probability measure $\mathbf{P}$ is uncertain for the DM. We introduce a {{set of dominated }}
probability measures ${{\mathcal{P}=}}\{Q^\mu,~\mu \in \mathbb{D}\}$ to describe the possible priors of the market.
$\mathbb{D}$ is the index set of the parameter $\mu$ and is supposed to be a separable metric space, and $\mathcal{B}(\mathbb{D})$ is the Borel $\sigma$-algebra of $\mathbb{D}$.
From the DM's perspective, {{for any fixed $\mu \in \mathbb{D}$}}, $Q^\mu$ is a possible candidate of the real probability measure $\mathbf{P}$. 
{{The DM's subjective belief over $\mathcal{P}$ is characterized by a SOD $F(\mu)$ on the measurable index space $(\mathbb{D},   \mathcal{B}(\mathbb{D}))$, which acts as a weighting scheme on $\mathcal{P}$ and represents the DM's subjective confidence level of a particular prior $Q^\mu$ to be the real probability measure.}}
The Lebesgue-Stieltjes measure $\mathbb{F}$ of $F(\mu)$ is a probability measure on the measurable index space $(\mathbb{D},   \mathcal{B}(\mathbb{D}))$, denoted by $\mathbb{F}(B)=\int_B\mathrm{d}F(\mu)$, for $B\in\mathcal{B}(\mathbb{D})$.
{{We impose the following assumption on $\mathcal{P}$ and $S^1$. 
\begin{assumption}\label{domin}
$\mathcal{P}=\{Q^\mu,~\mu \in \mathbb{D}\}$ satisfies the following conditions:

\begin{enumerate}
\item There exists an observable probability measure $\hat{Q}\ll\mathbf{P}$, s.t. $\mathcal{P}$ is dominated by $\hat{Q}$, i.e.,
$Q^\mu\ll\hat{Q}$, $\forall Q^\mu\in\mathcal{P}$.
\item $\forall A\in\mathcal{F}^S_T$, $Q^\mu(A)$ is a $\mathcal{B}(\mathbb{D})$-measurable function of $\mu$.
\item $S^1$ is a semi-martingale under $\hat{Q}$ and any $Q^\mu\in\mathcal{P}$.
\end{enumerate}
\end{assumption}}}
Note that $\hat{Q}$ does not necessarily belong to $\mathcal{P}$. 
Denote ${{\eta^\mu_T}}=\frac{\mathrm{d}Q^\mu}{\mathrm{d}{{\hat{Q}}}}|_{\mathcal{F}^S_t},$  which represents the Radon-Nikodym derivative of ${Q^\mu}$ with respect to the {{$\hat{Q}$}} on the filtration ${\mathcal{F}^S_t}$. 

\vskip 3pt
\vskip 3pt
Let $\pi=\{\pi_t,{0\leq t\leq T}\}$ be the money invested in {{$S^1$}}.
{We call a progressively measurable process $\pi$ integrable if $\pi$ satisfies the integrability condition} $\int_0^T\frac{\pi^2_t}{{{{S^1_t}^2}}}\mathrm{d}{{\!<\!S^1\!>_t}}<+\infty,~a.s.~\hat{Q}$,
where $<\!{{S^1}}\!>=\{{{<\!S^1\!>_t}},0\leq t\leq T\}$ is the quadratic variation of $S_1$, and {{the integrability condition also holds for any $Q^\mu\in\mathcal{P}$ as  $Q^\mu$ is dominated by $\hat{Q}$}}.
Then the DM's wealth process $X=\{{{X_t}}, 0\leq t\leq T\}$ with {an integrable} investment strategy $\pi$ satisfies the following stochastic differential equation (SDE):
 \begin{equation}\label{xg}
 \left\{
 \begin{array}{ll}
 \mathrm{d}{{X_t}}&=r({{X_t}}-{{\pi_t}})\mathrm{d}t+{{\pi_t\frac{\mathrm{d}S^1_t}{S^1_t}}},\\
 {{X_0}}&=x.
 \end{array}
 \right.
 \end{equation}
{
Moreover, the admissible set $\mathcal{V}[0,T]$ is defined by
\begin{equation*}
	\mathcal{V}[0,T]=\left\{\pi: \pi~\text{is progressively measurable and integrable;}~
	X~\text{defined by SDE~(\ref{xg}) is nonnegative}~a.s.~\hat{Q}
	\right\}.
\end{equation*}}
{{Denote $\mathfrak{X}(x)$ as the set of the wealth process $X$ satisfying SDE~(\ref{xg}) {with $\pi\in\mathcal{V}[0,T]$}.
}}
As $\mathcal{P}$ is dominated by ${{\hat{Q}}}$, ${{\hat{Q}}}$ is regarded as the reference measure of the financial market, and the models of assets and wealth process under ${{\hat{Q}}}$ are the reference models. 
{{We denote the set of 
local martingale measures by $\mathcal{M}$, and impose the other basic assumption on the financial market, which is very standard in the utility maximisation literature in mathematical finance and also imposed in \cite{Bartl21}.
\begin{assumption}\label{local}
For every $Q^\mu\in\mathcal{P}$, there exists at least one $Q\in\mathcal{M}$ such that $Q\ll Q^\mu$.
\end{assumption}
In particular, if the market is complete, i.e., $\mathcal{M}=\left\{Q\right\}$, we know $Q\approx\mathbf{P}$. Then by Assumptions \ref{domin} and \ref{local}, we know that for every $Q^\mu\in\mathcal{P}$, $Q\ll Q^\mu\ll\hat{Q}\ll\mathbf{P}$. Thus, in a complete market, Condition (1) of Assumption \ref{domin} and Assumption \ref{local} are equivalent to the condition that $Q^{\mu}\approx\mathbf{P}$, $\forall Q^\mu\in\mathcal{P}$, and 
any observable semi-martingale probability measure equivalent to $\mathbf{P}$ (or any one $Q^\mu\in\mathcal{P}$) can be chosen as $\hat{Q}$, the domination of $\mathcal{P}$. The results in the complete market will be shown in Sections 4 and 5 by an example of a Black-Scholes financial market, and we concentrate more on the KMM problem in an incomplete market in Sections 2 and 3.
}}



\subsection{{{\bf Optimization problem}}} \vskip 5pt
In order to describe the DM's ambiguity attitudes and  risk attitudes in the ambiguous market, we introduce the smooth ambiguity model proposed by 
\cite{7}, {{and the optimization problem under smooth ambiguity model in a general financial market in continuous time is still open.}} 
The objective function is given by
\begin{equation}\label{ut2}
\begin{array}{lll}
\Phi({{X^\pi_T}})&&= \int_{\mathbb{D}}\phi\left(\mathbf{E}^{Q^\mu}\left[U\left({{X^\pi_T}}\right)\right]\right)\mathrm{d}F(\mu) \\
&&=\int_{\mathbb{D}}\phi\left(\mathbf{E}^{{{\hat{Q}}}}\left[U({{X^\pi_T}}){{\eta^\mu_T}}\right]\right)\mathrm{d}F(\mu) ,
\end{array}
\end{equation}
where $U$ is a von Neumann-Morgenstern utility function depicting the DM's attitude towards financial risks and $\phi$ is a strictly increasing function depicting DM's ambiguity attitudes. The ambiguous yield's possible range $\mathbb{D}$ and SOD $F(\mu)$ characterize DM's subjective information of ambiguity. 
We suppose that $U$ {takes value in $[0,+\infty)$ and} satisfies the Inada conditions
 \[U'(x)>0, U'(0)=+\infty, U'(+\infty)=0, U''(x)<0,\]
 which means that $U$ is strictly increasing and concave.
 {
\begin{remark}
	If the utility function $U$ takes value in $[a,+\infty)$ for $a<0$, define $\tilde{U}(x)=U(x)-a$ and $\tilde{\phi}(x)=\phi(x+a)$, then we have the fact that $\tilde{U}$ takes value in $[0,+\infty)$ and $\tilde{\phi}\!\circ\!\tilde{U}=\phi\!\circ\!U$.
\end{remark}
}
 Different from the max-min ambiguity model in \cite{Gil89}, this optimization rule is concerned with the average performance and can separate the risk attitudes and ambiguity attitudes. For the ambiguity attitudes $\phi$, 
   we assume that 
   $\phi$ is strictly increasing and concave, i.e.,
   \[\phi'(x)>0, \phi''(x)<0.\]

The main purpose of this paper is to search the optimal strategy $\pi$  to maximize $\Phi({{X^\pi_T}})$ defined by  Eq.~(\ref{ut2}), i.e., {{solving}} the following problem:
\begin{equation}\label{p0}
\max\limits_{X\in\mathfrak{X}(x)}\left\{\Phi\left({{X^\pi_T}}\right)\right\}
\end{equation}

Most work concerning smooth ambiguity employ the recursive form in \cite{Klibanoff}. Numerical methods are applied to obtain the solution. We see that Problem~(\ref{p0}) contains two utility functions and two expectations, as such, it is time-inconsistent. {{Rather than using stochastic dynamic programming method on the recursive form, 
we aim to establish duality theorem for Problem~(\ref{p0}) and search the semi-explicit form of the pre-commitment solution under smooth ambiguity. 
As such, in our model, the DM is assumed to be pre-committed and does not update the subjective beliefs of $\mu$, i.e., we do not consider the updating of $\{Q^\mu,~\mu \in \mathbb{D}\}$ and $F$ with new information here.
}}

\setcounter{equation}{0}
\section{{ {\bf Solution of Problem~(\ref{p0})}}}
The {{objective function}} 
in Problem~(\ref{p0}) is very different from the {{classical}} maximization of  expected utility of terminal wealth (EUT) problem. In this section, we {{transform Problem~(\ref{p0}) into the combination of two kinds of one-fold expected utility optimization problems}} and solve Problem~(\ref{p0}) in the following two steps. In the first step, we consider the admissible set of the {{expected}} utilities 
under different priors.  Then we solve a multi-objective criteria problem{{ (which is  equivalent to a classical EUT problem under a weighted average probability measure)}} to find the efficient frontier of the admissible set. In the second step, {{it is enough}} 
to solve Problem~(\ref{p3}) 
because $\phi$ is strictly increasing, {{which is an one-fold expected utility optimization problem on the second-order space. Under Assumption \ref{Ass},}} we {{establish the bipolar relation in Theorem \ref{bipolarThm} 
and distorted duality theorem in Theorem \ref{duality}. Finally, under a further assumption that the asymptotic elasticity of the ambiguous attitude is less than 1, the solution to Problem~(\ref{p3}) is obtained.}} 
\vskip 5pt

\subsection{{{\bf  The efficient frontier}}}Because it is hard to obtain the optimal strategies on the whole admissible set  $\mathcal{V}[0,T]$ for Problem~(\ref{p0}) directly, we derive a smaller set of efficient strategies first. We introduce two sets here: $\mathbf{H}(x)$ is the set related with all admissible strategies, $\mathbf{B}(x)$ is a smaller set related with the efficient strategies. Our first step is to reduce from the whole admissible set $\mathbf{H}(x)$ to the efficient frontier $\mathbf{B}(x)$. Afterwards, we search the optimal strategy on the efficient frontier. First, we define the admissible set $\mathbf{H}(x)$ and the efficient frontier set $\mathbf{B}(x) \subset \mathbf{H}(x)$ with initial value $x$ for Problem~(\ref{p0}).
\vskip 10pt
\begin{definition}\label{D1}
{{For every $X\in\mathfrak{X}(x)$ and $\mu\in\mathbb{D}$, denote $\mathbf{b}(\mu,X)\triangleq\mathbf{E}^{Q^\mu}\left[U(X_T)\right]$}}. Then
the admissible set $\mathbf{H}(x)$ is 
defined by
\[\mathbf{H}(x)=\left\{ {{\mathbf{b}(\cdot,X):\mathbb{D}\rightarrow\mathbb{R}~
|~X\in\mathfrak{X}(x)}}\right\},\]
{which is a subset of $\mathbf{L}^0(\mathbb{D}, \mathcal{B}(\mathbb{D}), \mathbb{F})$, the the vector space of 	(equivalence classes of) real-valued measurable functions defined on $(\mathbb{D}, \mathcal{B}(\mathbb{D}), \mathbb{F})$ equipped with the topology of convergence in measure.}
 $\mathbf{B}(x)$ is defined as a subset of $\mathbf{H}(x)$. {{For $\hat{X}\in\mathfrak{X}(x)$, function $\mathbf{b}(\cdot,\hat{X})
\in \mathbf{B}(x) $}}
if and only if  for any ${{X\in\mathfrak{X}(x)}}$,
{{\[\mathbf{b}(\cdot,X)-\mathbf{b}(\cdot,\hat{X})\notin {\mathbf{L}^0_+(\mathbb{D}, \mathcal{B}(\mathbb{D}), \mathbb{F})}\setminus\{0\},\]}}
where ${\mathbf{L}^0_+(\mathbb{D}, \mathcal{B}(\mathbb{D}), \mathbb{F})}\triangleq\{\mathbf{d}{\in\mathbf{L}^0(\mathbb{D}, \mathcal{B}(\mathbb{D}), \mathbb{F}),\mathbf{d}\geq 0}
\}$.
\end{definition}
\vskip 10pt

In fact, if ${{\hat{X}}}$ attains the maximum of $\Phi({{X_T}})$ in Problem~(\ref{p0}), then {{$\mathbf{b}(\cdot,\hat{X})\in \mathbf{B}(x)$.}} If {{$\mathbf{b}(\cdot,\hat{X})\notin \mathbf{B}(x)$}}, then there exists some {{$\tilde{X}\in\mathfrak{X}(x)$ satisfying $\mathbf{b}(\cdot,\tilde{X})\in \mathbf{B}(x)$}} such that
{{\[ \mathbf{b}(\cdot,\tilde{X})-\mathbf{b}(\cdot,\hat{X})\in {\mathbf{L}^0_+(\mathbb{D}, \mathcal{B}(\mathbb{D}), \mathbb{F})}\setminus\{0\}.\]}}
However, $\phi(x)$ is strictly increasing. As such,
\[\int_{\mathbb{D}}\phi\left(\mathbf{E}^{Q^\mu}\left[U({{\tilde{X}_T}})\right]\right)\mathrm{d}F(\mu) > \int_{\mathbb{D}}\phi\left(\mathbf{E}^{Q^\mu}\left[U({{\hat{X}_T}})\right]\right)\mathrm{d}F(\mu),\]  which contradicts with the fact that ${{\hat{X}_T}}$ attains the maximum of $\Phi({{X_T}})$. Therefore, it is necessary and enough for us to study the efficient frontier set $\mathbf{B}(x)$. Next we characterize the set $\mathbf{B}(x)$. The efficient frontier $\mathbf{B}(x)$ can be derived by a multi-goal optimization problem: $\max\limits_{{{X\in\mathfrak{X}(x)}}}\left\{\mathbf{b}(\mu,X),\mu\in \mathbb{D}\right\}$. 
In \cite{Zhou}, the mean-variance problem is two-goal and they transform the problem into an equivalent one maximizing the mean minus the weighted variance. In our work, we also expect to transform the  multi-goal optimization problem to an equivalent one with only one goal. For this purpose, we investigate the relationship between $\mathbf{B}(x)$ and an equivalent optimization problem 
in the following theorem.
\vskip 10pt
\begin{theorem}\label{TT1}
For {{$\hat{X}\in\mathfrak{X}(x)$,}}  ${{\mathbf{b}(\cdot,\hat{X})}}\in \mathbf{B}(x)$ if and only if there exists a function $\lambda\in {\mathbf{L}^0_+(\mathbb{D}, \mathcal{B}(\mathbb{D}), \mathbb{F})}$ satisfying $ <\lambda,\mathbb{I}_{\mathbb{D}}>=1$ such that
\[{{\hat{X}}} = \arg\max\limits_{{{X\in\mathfrak{X}(x)}}}\left\{<{{\mathbf{b}(\cdot,X)}},\lambda(\cdot)>\right\},\]
where 
$\mathbb{I}_{\mathbb{D}}$ is the indicator function of $\mathbb{D}$, and {$<f_1,f_2>\triangleq\int_{\mathbb{D}}f_1 f_2\mathrm{d}\mathbb{F}$, for $f_1$, $f_2\in\mathbf{L}^0_+(\mathbb{D}, \mathcal{B}(\mathbb{D}), \mathbb{F})$. 
}
\end{theorem}
\begin{proof}
	See Appendix \ref{proof of TT1}.
\end{proof}
\begin{remark}
{$<\cdot,\cdot>$ is not a scalar product in the usual sense of the word as it may assume the value $+\infty$. But the expression  is a well-defined element of $[0,+\infty]$ and the application $(\cdot,\cdot)\rightarrow<\cdot,\cdot>$ has, with necessary modifications, the obvious properties of a bilinear function.
}
\end{remark}

\vskip 10pt
{{
Theorem \ref{TT1} shows that in order to obtain the efficient frontier $\mathbf{B}(x)$, we only need to solve the following problem without function $\phi(\cdot)$ first:
{{
\begin{equation}\label{pr22}
J(x,\lambda)=\max\limits_{{{X\in\mathfrak{X}(x)}}}\left\{<\lambda(\cdot),\mathbf{E}^{Q^{\bigcdot}}\left[U({{X_T}})\right]>\right\}.
\end{equation}
For every fixed $\lambda\in{\mathbf{L}^0_+(\mathbb{D}, \mathcal{B}(\mathbb{D}), \mathbb{F})}$ satisfying $<\lambda,\mathbb{I}_{\mathbb{D}}>=1$, 
we know 
\begin{equation}\nonumber
\begin{split}
\mathbf{E}^{\hat{Q}}\left[\int_\mathbb{D}\lambda(\mu)\eta^\mu_T\mathrm{d}\mathbb{F}\right]&=\int_\mathbb{D}\lambda(\mu)\mathbf{E}^{\hat{Q}}\left[\eta^\mu_T\right]\mathrm{d}\mathbb{F}\\
&=\int_\mathbb{D}\lambda(\mu)\mathbf{E}^{Q^{\mu}}\left[~1~\right]\mathrm{d}\mathbb{F}\\
&=1.
\end{split}
\end{equation}
Denote
\begin{equation}\nonumber
{}^\lambda\eta_T=\int_\mathbb{D}\lambda(\mu)\eta^\mu_T\mathrm{d}\mathbb{F},
\end{equation}
with ${}^\lambda Q$ satisfying
\begin{equation}\nonumber
\frac{\mathrm{d}{}^\lambda Q}{\mathrm{d}Q^{\mu_0}}|_{\mathcal{F}_T^S}={}^\lambda\eta_T.
\end{equation}
Then we have
\begin{equation}\nonumber
\begin{split}
<\lambda(\cdot),\mathbf{E}^{Q^{\bigcdot}}\left[U({{X_T}})\right]>&=\int_\mathbb{D}\lambda(\mu)\mathbf{E}^{\hat{Q}}\left[U(X_T)\right]\mathrm{d}\mathbb{F}\\
&=\mathbf{E}^{\hat{Q}}\left[\int_\mathbb{D}\lambda(\mu)U(X_T)\eta^\mu_T\mathrm{d}\mathbb{F}\right]\\
&=\mathbf{E}^{\hat{Q}}\left[U(X_T)\int_\mathbb{D}\lambda(\mu)\eta^\mu_T\mathrm{d}\mathbb{F}\right]\\
&=\mathbf{E}^{\hat{Q}}\left[U(X_T){}^\lambda\eta_T\right]\\
&=\mathbf{E}^{{}^\lambda Q}\left[U(X_T)\right].
\end{split}
\end{equation}
Problem~(\ref{pr22}) can be written as
\begin{equation}\label{pr2}
J(x,\lambda)=\max\limits_{X\in\mathfrak{X}(x)}\left\{\mathbf{E}^{{}^\lambda Q}\left[U(X_T)\right]\right\},
\end{equation}
which is a classical EUT problem for fixed $\lambda$.

Based on Assumption \ref{local}, for every $Q^\mu\in\mathcal{P}$, $\exists Q^\mu_*\in\mathcal{M}$ such that $Q^\mu_*\ll Q^\mu$. Then, $\forall A\in\mathcal{F}^S_T$, denote
\begin{equation}\nonumber
{}^\lambda Q^*(A)=\int_{\mathbb{D}}\lambda(\mu)Q^\mu_*(A)\mathrm{d}\mathbb{F},
\end{equation}
and ${}^\lambda Q^*$ is also a local martingale measure. For any $A\in\mathcal{F}^S_T$ satisfying ${}^\lambda Q(A)=0$, as $\lambda(\mu)\geq 0$ and $Q^\mu(A)\geq 0$ for any $\mu\in\mathbb{D}$, we have
\begin{equation}\nonumber
Q^\mu(A)=0~\text{on}~\Lambda,~\text{a.s., }~\mathbb{F},
\end{equation}
where $\Lambda$ is the support set of $\lambda$. As $Q^\mu_*\ll Q^\mu$ for all $\mu\in\mathbb{D}$, we have
\begin{equation}\nonumber
Q^\mu_*(A)=0~\text{on}~\Lambda,~\text{a.s.}~\mathbb{F}.
\end{equation}
As such, ${}^\lambda Q^*(A)=0$, which means ${}^\lambda Q^*\ll {}^\lambda Q$, i.e., there exists a local martingale measure for Problem~(\ref{pr2}). {Based on \cite{25}, $\tilde{\mathcal{D}}$ is the subset of $\mathcal{D}$ consisting of the functions $h$ of the form $h = \frac{\mathrm{d}Q}{\mathrm{d}\mathbf{P}}$, for some $Q\in\mathcal{M}^e$, where $\mathcal{M}^e$ is the set of equivalent martingale measures. Similarly, denote $\bar{\mathcal{D}}$ the subset of $\mathcal{D}$ consisting of the functions $h$ of the form $h = \frac{\mathrm{d}Q}{\mathrm{d}\mathbf{P}}$, for some $Q\in\mathcal{M}$. Note that
\begin{equation}\nonumber
	\mathcal{D}\subset\bar{\mathcal{D}}^{00}\subset\mathcal{D}^{00},
\end{equation}
where $\mathcal{D}^{00}=\left(\mathcal{D}^0\right)^0$, and $\mathcal{D}^0$ is the polar of $\mathcal{D}$ defined by
\begin{equation}\nonumber
	\mathcal{D}^0\triangleq\left\{f\in\mathbf{L}^0_+(\Omega,\mathcal{F}^S,{}^\lambda Q):\mathbf{E}^{{}^\lambda Q}[fg]\leq 1, \forall g\in \mathcal{D}\right\}.
\end{equation}
Based on the results in Section 4 in \cite{25}, we have
\begin{equation}\nonumber
	\mathcal{D}=\bar{\mathcal{D}}^{00}=\mathcal{D}^{00},
\end{equation}
which means that the results in \cite{25} still hold if there exists at least one local martingale measure.
}

Based on {Theorem 2.2} in \cite{25}, for any fixed $\lambda\in{\mathbf{L}^0_+(\mathbb{D}, \mathcal{B}(\mathbb{D}), \mathbb{F})}$ satisfying $<\lambda,\mathbb{I}_{\mathbb{D}}>=1$, $J(x,\lambda)$ is a continuously differentiable strictly concave function of $x$ satisfying the Inada condition, and the terminal wealth $\hat{X}_T$ of the point on the efficient frontier corresponding to $\lambda$ is given by
\begin{equation}
\hat{X}^\lambda_T=I(\hat{Y}^\lambda_T)e^{rT},
\end{equation}
where $I(\cdot)=(U')^{-1}(\cdot)$, and $\hat{Y}^\lambda_T$ is the solution of the following problem
\begin{equation}\label{infY}
\inf\limits_{Y\in\mathscr{Y}(y)}\mathbf{E}^{{}^\lambda Q}[V(Y_T)],
\end{equation}
where $y=J_x(x,\lambda)$, $V(y)=\sup\limits_{x>0}\left[U(x)-xy\right]$, for $y>0$, and
\begin{equation}\nonumber
\mathscr{Y}(y;\lambda)=\left\{Y\!\geq 0:Y_0=y~\text{and}~\{X_tY_te^{-rt}\!,0\!\leq\!t\!\leq\!T\}~\text{is a supermartingale under}~{}^\lambda Q~\text{for all}~X\!\in\!\mathfrak{X}(1)\right\}.
\end{equation}

}}

We have shown that the efficient claim ${{\hat{X}^\lambda_T}}$ exists for any $\lambda\in{\mathbf{L}^0_+(\mathbb{D}, \mathcal{B}(\mathbb{D}), \mathbb{F})}$ satisfying $<\lambda,\mathbb{I}_{\mathbb{D}}>=1$. 
Next, we replicate ${{\hat{X}^\lambda_T}}$ and derive the strategy $\hat{\pi}^\lambda$ corresponding to ${{\hat{X}^\lambda_T}}$.
Based on \cite{25} (assertion (ii) of Theorem 2.2), $\left\{\hat{X}^\lambda_t\hat{Y}^\lambda_t e^{-rt},0\leq t\leq T\right\}$ is a uniformly integrable martingale under ${}^\lambda Q$ and $\hat{X}^\lambda$ is strictly positive. Denote a probability measure ${}^\lambda Q'$ by
\begin{equation}\label{mart}
\frac{\mathrm{d}{}^\lambda Q'}{\mathrm{d}{}^\lambda Q}|_{\mathcal{F}^S_T}=\frac{\hat{X}^\lambda_T\hat{Y}^\lambda_T}{xye^{rT}}.
\end{equation}
Based on \cite{DS95}, the discounted processes of $1/\hat{X}^\lambda$ and $S^1/\hat{X}^\lambda$ are both local martingales under ${}^\lambda Q'$.
As such, by the martingale representation theorem, there exists a unique stochastic process $\{\hat{f}^\lambda_t, 0\leq t\leq T  \}$ such that
\begin{equation}\label{XX2}
\begin{split}
\frac{e^{rt}}{\hat{X}^\lambda_t} 
&= \frac{1}{x} + \int_0^t \hat{f}^\lambda_u\mathrm{d}\left(\frac{S^1_u e^{ru}}{\hat{X}^\lambda_u}\right).
\end{split}
\end{equation}
Based on SDE~(\ref{xg}), we have
\begin{equation}\label{XX3}
\left\{
\begin{split}
&\mathrm{d}\left(\frac{e^{rt}}{\hat{X}_t}\right)=\frac{r\hat{\pi}_t e^{rt}}{\hat{X}_t^2}\mathrm{d}t+\frac{e^{rt}}{{\hat{X}_t}^3}\frac{\hat{\pi}_t}{S^1_t}\frac{\mathrm{d}[S^1]_t}{S^1_t}-\frac{\hat{\pi}_t e^{rt}}{{\hat{X}_t}^2}\mathrm{d}S^1_t,\\
&\mathrm{d}\left(\frac{S^1_t e^{rt}}{\hat{X}_t}\right)=\frac{r\hat{\pi}_t e^{rt}}{{\hat{X}_t}^2}S^1_t\mathrm{d}t+\frac{\hat{\pi}_t e^{rt}}{{\hat{X}_t}^2}\left(\frac{\hat{\pi}_t}{4\hat{X}_t}-1\right)\frac{\mathrm{d}[S^1]_t}{S^1_t}+\frac{e^{rt}}{{\hat{X}_t}}\left(\frac{\hat{\pi}_t}{\hat{X}_t}-1\right)\mathrm{d}S^1_t.
\end{split}
\right.
\end{equation}
Comparing the diffusion parts in Eq.~(\ref{XX2}) and
Eq.~(\ref{XX3}), we obtain the efficient strategy at time $t$:
\begin{equation*}
\hat{\pi}^\lambda_t=\frac{\hat{f}^\lambda_t}{\hat{f}^\lambda_t-1} \hat{X}^\lambda_t.
\end{equation*}

}}

We have proved that every efficient claim ${{\hat{X}^\lambda_T}}$ and the related  strategy {{$\hat{\pi}^\lambda$}} can be obtained for all $\lambda\in{\mathbf{L}^0_+(\mathbb{D}, \mathcal{B}(\mathbb{D}), \mathbb{F})}$ satisfying $<\lambda,\mathbb{I}_{\mathbb{D}}>=1$. Varying $\lambda\in{\mathbf{L}^0_+(\mathbb{D}, \mathcal{B}(\mathbb{D}), \mathbb{F})}$, we can obtain the efficient frontier $\mathbf{B}(x)$. As $\phi(\cdot)$ is an increasing function, the solution of Problem~(\ref{p0}) is obviously on the efficient frontier. Thus, we only need to search the optimal strategy of Problem~(\ref{p0}) on the efficient frontier $\mathbf{B}(x)$ later.

\subsection{{{\bf Duality of Problem~(\ref{p0}) on the efficient frontier}}}
After characterizing the efficient frontier $\mathbf{B}(x)$, we now solve Problem~(\ref{p0}) on $\mathbf{B}(x)$ by extending the duality theorem. Because $\phi(\cdot)$ is strictly increasing, Problem~(\ref{p0}) is equivalent to the following problem
\begin{equation}\label{p2}
u(x)=\mathop {\max }\limits_{{{\mathbf{b}(\cdot,\hat{X})}}\in \mathbf{B}(x)}\left\{\int_{\mathbb{D}}\phi\left({{\mathbf{b}(\mu,\hat{X})}}\right)\mathrm{d}F(\mu)\right\},
\end{equation}
where $\mathbf{B}(x)$  is the efficient frontier with initial wealth $x$.
{{In \cite{Bartl21}, 
the duality theory is derived for the robust utility maximization problem, 
in which the duality theorem and the bipolar relation are established on the sets of super-hedgeable claims and separating measures. However, in our model, the DM is no longer only concerned with the expected utility under the worst case scenario, but the nonlinear weighted average of the expected utilities w.r.t. different priors. As such, we aim to establish the duality theorem and the bipolar relation on the sets of expected utilities w.r.t. different priors and weight functions rather than directly on the sets of the claims and probability measures. 
}}
\vskip 5pt
Noting that Problem~(\ref{p2}) is an expected utility optimization problem on $\mathbf{B}(x)$,  because there are more than one weight function $\lambda$ satisfying $\lambda\in{\mathbf{L}^0_+(\mathbb{D}, \mathcal{B}(\mathbb{D}), \mathbb{F})}$ and $<\lambda,\mathbb{I}_{\mathbb{D}}>=1$, we first compare Problem~(\ref{p2}) with the classical EUT problem in an incomplete market. Then we derive the dual admissible set
and {{establish the bipolar relation and the dual theorem to solve Problem~(\ref{p2}).}}

Based on Theorem \ref{TT1}, 
every $\lambda\!\in\!{\mathbf{L}^0_+(\mathbb{D}, \mathcal{B}(\mathbb{D}), \mathbb{F})}$ satisfying $<\lambda,\mathbb{I}_{\mathbb{D}}>=1$ is related with a unique point ${{\mathbf{b}(\cdot,\hat{X}^\lambda)}}\in\mathbf{B}(x)$, {{which solves Problem~(\ref{pr2}).}} Let $\mathbf{S}=\{\lambda\in{\mathbf{L}^0_+(\mathbb{D}, \mathcal{B}(\mathbb{D}), \mathbb{F})}, <\lambda,\mathbb{I}_{\mathbb{D}}>=1\}$. 
{{Recall the fact that, for any $\lambda\in\mathbf{S}$,}}
	\begin{equation}\nonumber
		J(x,\lambda)=<{{\mathbf{b}(\cdot,\hat{X}^\lambda)}},\lambda(\cdot)>.
	\end{equation}
	Then, based on Theorem \ref{TT1}, for any 
$\lambda\in\mathbf{S}$ and any ${{\mathbf{b}(\cdot,\hat{X})}}\in\mathbf{B}(x)$ {{{(not necessarily the point related to $\lambda$)}}, we have
	\begin{equation}\label{budget}
		\begin{split}
			\int_{\mathbb{D}}{{\mathbf{b}(\mu,\hat{X})}}\lambda(\mu)\mathrm{d}F(\mu)&=<{{\mathbf{b}(\cdot,\hat{X})}},\lambda(\cdot)>\\
			&\leq<{{\mathbf{b}(\cdot,\hat{X}^\lambda)}},\lambda(\cdot)>\\
			&=J(x,\lambda).
		\end{split}
	\end{equation}
The form of Eq.~(\ref{budget}) is similar to the budget constraint of the classical EUT problem.
	To be more clear, the point $\mathbf{b}$ on the efficient frontier corresponds to the terminal wealth ${{X_T}}$ in classical EUT problem; the weight function $\lambda(\cdot)$ corresponds to ${{\xi_T}}$, the terminal variable of the density process of the equivalent local martingale measure in incomplete market in classical EUT problem because $\lambda(\cdot)$ is not unique; the integration $\int_{\mathbb{D}}\cdot~\mathrm{d}F(\mu)$ corresponds to the expectation $\mathbf{E}^{\mathbf{P}}[~\cdot~]$ in classical EUT problem. The main difference between Eq.~(\ref{budget}) and the budget constraint of classical EUT problem is that the constraint becomes $J(x,\lambda)$, no longer a value of $xe^{rT}$ only depending on $x$ in the classical EUT problem. The dependence of $J(x,\lambda)$  on $\lambda$ makes it impossible to derive {{the bipolar relation}} directly as the classical EUT problem.
\vskip 5pt
If $J(x,\lambda)$ is independent of $\lambda$, Eq.~(\ref{budget}) becomes the distorted budget constraint of the efficient frontier $\mathbf{B}(x)$, and $J(x)$ is the constraint value corresponding to $xe^{rT}$ in classical EUT problem. As such, the structures of $\mathbf{B}(x)$ and the admissible set $\mathfrak{X}(x)$ in incomplete market in 
\cite{25} are similar. Then Problem~(\ref{p2}) has the same form as the EUT problem in an incomplete market in 
\cite{25}, which is
	\begin{equation}\nonumber
		\max\limits_{X\in\mathfrak{X}(x)}\left\{\mathbf{E}^{\mathbf{P}}\left[U({{X_T}})\right]\right\}.
	\end{equation}
The original optimization problem~(\ref{p0}) is a combination of two utility functions. Using the above arguments, the optimal terminal claim can be obtained on $\mathbf{B}(x)$ by studying Problem~(\ref{p2}). Problem~(\ref{p2}) is similar to the problem under an incomplete market in 
\cite{25} {{with a distorted initial value. However, it is fairly unrealistic that $J(x,\lambda)$ is independent of $\lambda$.}}
\vskip 5pt
In most cases, $J(x,\lambda)$ depends on $\lambda$,  see Section 4 for details. If $J(x,\lambda)$ depends on $\lambda$, it is impossible to derive the {{classical}} duality theorem for Problem~(\ref{p2}). However, if $x$ and $\lambda$ can be separated in $J(x,\lambda)$, we {{are able to}} extend the dual method in 
\cite{25} {{and establish a distorted duality theorem}} for Problem~(\ref{p2}). 
As such, we assume that $x$ and $\lambda$ can be separated in the following assumption. And we will see in Section 4 that most of the utility functions satisfy the assumption.

\begin{assumption}\label{Ass}
For any initial value $x\!>\!0$ and non-negative  weight function {{$\lambda\in\mathbf{S}$,}} 
there exist {a function $h(x)$ and a 
functional $\rho(\lambda)$} such that
\begin{equation}\label{asp}
\begin{split}
J(x,\lambda)=~<{{\mathbf{b}(\cdot,\hat{X}^\lambda)}}
,\lambda(\cdot)>~=~h(x)\rho(\lambda).
\end{split}
\end{equation}
\end{assumption}
{{If $h(x)$ and $\rho(\lambda)$ can be arbitrage chosen, the uniqueness  may not exist because the pair $\{\tilde{h}\triangleq kh,\tilde{\rho}\triangleq\frac{1}{k}\rho\}$ also satisfies Assumption \ref{Ass} for any $k>0$, $k\neq 1$. As such, we suppose that every $h(x)$ mentioned below satisfies $h(1)=U(e^{rT})$.


Noting that $xS^0=\{xe^{rt},0\leq t\leq T\}\in\mathfrak{X}(x)$, we know $\mathbf{b}(\cdot,xS^0)\in\mathbf{H}(x)$ and $\mathbf{b}(\mu,xS^0)=U(xe^{rT})$, $\forall \mu\in\mathbb{D}$. Based on Eq.~(\ref{asp}), for any 
 weight function $\lambda\in \mathbf{S}$, the functional $\rho$ is supposed to satisfy
\begin{equation}\nonumber
\begin{split}
\rho(\lambda)&=\frac{<\mathbf{b}(\cdot,\hat{X}^\lambda),\lambda(\cdot)>}{h(x)}\\
&\geq\frac{<\mathbf{b}(\cdot,xS^0),\lambda(\cdot)>}{h(x)}\\
&=\frac{U(xe^{rT})}{h(x)}.
\end{split}
\end{equation}
Let $x=1$, and we have
\begin{equation}\nonumber
\rho(\lambda)\geq 1,~\forall \lambda\in\mathbf{S},
\end{equation}
and the equality holds if and only if $\mathbf{b}(\cdot,S^0)\in\mathbf{B}(1)$ and $S^0=\hat{X}^\lambda$.
}}

In Assumption \ref{Ass}, if $\lambda$ is divided by $\rho(\lambda)$, {{for any $\mathbf{b}(\cdot,\hat{X})\in\mathbf{B}(x)$}}, we have
\begin{equation}\label{truebudget}
\begin{split}
&\int_{\mathbb{D}}{{\mathbf{b}(\mu,\hat{X})}}
\frac{\lambda(\mu)}{\rho(\lambda)}\mathrm{d}F(\mu)
=<{{\mathbf{b}(\cdot,\hat{X})}},\frac{\lambda(\cdot)}{\rho(\lambda)}>\\
&\leq<{{\mathbf{b}(\cdot,\hat{X}^\lambda)}},\frac{\lambda(\cdot)}{\rho(\lambda)}>
=h(x).
\end{split}
\end{equation}
As such, Eq.~(\ref{truebudget}) remains the form of the distorted budget constraint, and $h(x)$ becomes the distorted value of  the new budget constraint corresponding to $xe^{rT}$ in  a classical EUT problem, {{which helps us derive the dual admissible set and the bipolar relation}}.

{{Before the definition of the dual admissible set, we first present some properties of  $h$ and $\rho$. For the distortion function $h$, recall that for any fixed $\lambda\in\mathbf{S}$, $J(x,\lambda)$ is a continuously differentiable strictly concave function of $x$ and $J(x,\lambda)=h(x)\rho(\lambda)$. As such, $h$ is continuously differentiable and strictly concave and satisfies the Inada condition. As a result, both $h$ and $h'$ are invertible on $(0,\infty)$. 
}} Besides, in some particular cases, $h(x)$ is just {{equal}} to $U(xe^{rT})$, which is shown in the following proposition.
\vskip 10pt
\begin{proposition}\label{Prop1}
{{If $\mathcal{P}\cap\mathcal{M}\neq \emptyset$, i.e., there exists a local martingale measure $Q^*$ satisfying $Q^*\in\mathcal{P}$ and Assumption \ref{Ass} holds, then $h(x)=U(xe^{rT})$, $\forall x> 0$.}}
\end{proposition}
\begin{proof}
	See Appendix \ref{proof of Prop1}.
\end{proof}

\begin{remark}
The condition {{$\mathcal{P}\cap\mathcal{M}\neq \emptyset$}} in Proposition \ref{Prop1} holds in the examples  given in Section 4.
\end{remark}
\vskip 10pt
{{For the functional $\rho$, it is currently defined on $\mathbf{S}$, the set of normalized non-negative weight functions. However, Problem~(\ref{pr2}) and Eq.~(\ref{asp}) are also well-defined for $\tilde{\lambda}=k\lambda$ for $k>0$ and $\lambda\in\mathbf{S}$. 
It is obvious that $\hat{X}^{\tilde{\lambda}}=\hat{X}^\lambda$, then $\rho$ can be defined on ${\mathbf{L}^0_+(\mathbb{D}, \mathcal{B}(\mathbb{D}), \mathbb{F})}$ as follows
\begin{equation}\nonumber
\begin{split}
\rho(\tilde{\lambda})&=\frac{<\mathbf{b}(\cdot,\hat{X}^{\tilde{\lambda}}),\tilde{\lambda}(\cdot)>}{h(x)}\\
&=\frac{<\mathbf{b}(\cdot,\hat{X}^\lambda),k\lambda(\cdot)>}{h(x)}\\
&=k\rho(\lambda).
\end{split}
\end{equation}
As such, for any $\tilde{\lambda}\in{\mathbf{L}^0_+(\mathbb{D}, \mathcal{B}(\mathbb{D}), \mathbb{F})}$, $\frac{\tilde{\lambda}}{\rho(\tilde{\lambda})}=\frac{k\lambda}{k\rho(\lambda)}=\frac{\lambda}{\rho(\lambda)}$ and $k=<\tilde{\lambda},\mathbb{I}_{\mathbb{D}}>$, which means that $\rho$ is homogeneous. For any $\lambda_1$, $\lambda_2\in\mathbf{S}$ and $\alpha\in(0,1)$, denote $\bar{\lambda}=\alpha\lambda_1+(1-\alpha)\lambda_2$ we have
\begin{equation}\nonumber
\begin{split}
h(x)\rho(\bar{\lambda})&=
<\mathbf{b}(\cdot,\hat{X}^{\bar{\lambda}}),\bar{\lambda}>\\
&=\alpha<\mathbf{b}(\cdot,\hat{X}^{\bar{\lambda}}),\lambda_1>+(1-\alpha)<\mathbf{b}(\cdot,\hat{X}^{\bar{\lambda}}),\lambda_2>\\
&\leq\alpha<\mathbf{b}(\cdot,\hat{X}^{\lambda_1}),\lambda_1>+(1-\alpha)<\mathbf{b}(\cdot,\hat{X}^{\lambda_2}),\lambda_2>\\
&=\alpha h(x)\rho(\lambda_1)+(1-\alpha)h(x)\rho(\lambda_2).
\end{split}
\end{equation}
Therefore, $\rho\left(\alpha\lambda_1+(1-\alpha)\lambda_2\right)\leq\alpha\rho(\lambda_1)+(1-\alpha)\rho(\lambda_2)$, which means that $\rho$ is convex.
}}

Assumption \ref{Ass} is necessary to derive the dual admissible set and the bipolar relation. 
Otherwise, the dual method cannot be applied. However, we later show that the assumption holds when the utility function $U$ is either CARA function or HARA function (including CRRA function except for $\log(x)$).
\vskip 5pt
In the following, we define the dual admissible set {{and derive the bipolar relation.}}
Based on Assumption \ref{Ass} and Eq.~(\ref{truebudget}), we first  define the dual admissible set $\mathbf{G}(y)$ 
for $y>0$.
\vskip 10pt
\begin{definition}\label{DD2}
\begin{equation}\nonumber
\mathbf{G}(y)\!=\!\left\{g\!\in\! {\mathbf{L}^0_+(\mathbb{D},\mathcal{B}(\mathbb{D}),\mathbb{F})}
:\exists \lambda\!\in\!\mathbf{S} 
~\text{s.t.}~
g(\mu)\!\leq\!\frac{y\lambda(\mu)}{\rho(\lambda)},~\forall\mu\in\mathbb{D}\right\}\!.
\end{equation}
\end{definition}

{
Define $\mathbf{D}=\mathbf{L}^0_+(\mathbb{D}, \mathcal{B}(\mathbb{D}), \mathbb{F})$, and
\begin{equation}\nonumber
	\mathbf{B}(x)-\mathbf{D}=\left\{f\in\mathbf{L}^0_+(\mathbb{D}, \mathcal{B}(\mathbb{D}), \mathbb{F}):\exists \mathbf{b}(\cdot,\hat{X})\in\mathbf{B}(x),~\text{s.t.}~f(\mu)\leq\mathbf{b}(\mu,\hat{X}),~\forall\mu\in\mathbb{D}\right\}.
\end{equation}
}
Based on Assumption \ref{Ass}, it follows that, for $f\in \mathbf{B}(x)-\mathbf{D}, g\in\mathbf{G}(y)$,
\begin{eqnarray*}\nonumber
\int_{\mathbb{D}}{{fg~\mathrm{d}\mathbb{F}}}
&&=<f(\cdot),g(\cdot)>\\
&&\leq<f(\cdot),\frac{y\lambda(\cdot)}{\rho(\lambda)}>\\
&&\leq<{{\mathbf{b}(\cdot,\hat{X}^\lambda)}},\frac{y\lambda(\cdot)}{\rho(\lambda)}>\\
&&=h(x)y.
\end{eqnarray*}
\vskip 4pt
Summarizing the former results, we obtain the following relation
 between $\mathbf{B}(x)-\mathbf{D}$ and $\mathbf{G}(y)$.
\vskip 4pt
\begin{proposition}\label{set}
For any fixed $x>0,y>0$,
\begin{equation}\label{main}
\begin{split}
&f\in\mathbf{B}(x)-\mathbf{D}~\text{iff}~\int_{\mathbb{D}}{{fg~\mathrm{d}\mathbb{F}}}\leq h(x)y ~\text{ for all }~ g\in \mathbf{G}(y),\\
&g\in\mathbf{G}(y)~\text{iff}~\int_{\mathbb{D}}{{fg~\mathrm{d}\mathbb{F}}}\leq h(x)y ~\text{ for all }~ f\in\mathbf{B}(x)-\mathbf{D}.
\end{split}
\end{equation}
\end{proposition}
\begin{proof}
	See Appendix \ref{proof of set}.
\end{proof}
{{The relation (\ref{main}) is close to the bipolar relation and the difference is the dependence of $\mathbf{B}(x)-\mathbf{D}$ on $x$. In order to derive the standard bipolar relation, denote $\mathbf{G}=\mathbf{G}(1)$ and its polar by
\begin{equation}\nonumber
\mathbf{B}=\left\{f\in{\mathbf{L}^0_+(\mathbb{D},\mathcal{B}(\mathbb{D}),\mathbb{F})}:\int_\mathbb{D}fg~\mathrm{d}\mathbb{F}\leq 1~\text{for all}~g\in\mathbf{G}\right\}.
\end{equation}
Based on Proposition \ref{set}, we have
\begin{equation}\label{setdistort}
\left\{
\begin{split}
&\mathbf{B}(x)-\mathbf{D}=h(x)\mathbf{B},~\forall x>0,\\
&\mathbf{G}(y)=y\mathbf{G},~\forall y>0.
\end{split}
\right.
\end{equation}
Eq.~(\ref{setdistort}) shows that $\mathbf{B}(x)$ no longer  depends linearly on $x$ while depends on the distortion function $h(x)$.
Especially, 
$\mathbf{B}=\mathbf{B}(h^{-1}(1))-\mathbf{D}$, where $h^{-1}$ is the right inverse of $h$.

Recall $h(1)=U(e^{rT})$, we know $\mathbf{B}=\frac{1}{U(e^{rT})}\mathbf{B}(1)-\mathbf{D}$, and the bipolar relation is directly obtained and shown 
in the following theorem.
\begin{theorem}[\bf{bipolar relation}]\label{bipolarThm}
\begin{equation}\label{bipolar}
\begin{split}
&f\in\frac{1}{U(e^{rT})}\mathbf{B}(1)-\mathbf{D}~\text{iff}~\int_{\mathbb{D}}fg~\mathrm{d}\mathbb{F}\leq 1 ~\text{ for all }~ g\in \mathbf{G},\\
&g\in\mathbf{G}~\text{iff}~\int_{\mathbb{D}}fg~\mathrm{d}\mathbb{F}\leq 1~\text{ for all }~ f\in\frac{1}{U(e^{rT})}\mathbf{B}(1)-\mathbf{D}.
\end{split}
\end{equation}
\end{theorem}
}}
\vskip 10pt

As $\mathbf{B}(x)-\mathbf{D}$ 
 has the relation with $\mathbf{G}(y)$ in Proposition \ref{set}, we consider 
the following problem on $\mathbf{B}(x)-\mathbf{D}$, which is equivalent to Problem~(\ref{p2}).
\begin{equation}\label{p3}
u(x)=\max\limits_{f \in \mathbf{B}(x)-\mathbf{D}}\left\{ \int_{\mathbb{D}}{{\phi(f)\mathrm{d}\mathbb{F}}}
\right\}.
\end{equation}
Denote by $\hat{f}(x)$ the solution to Problem~(\ref{p3}). 
\vskip 4pt
Using the duality relation between  $\mathbf{B}(x)-\mathbf{D}$ and $\mathbf{G}(y)$ in Proposition \ref{set}, we study the dual problem of Problem~(\ref{p3}): 
\begin{equation}\label{D2}
v(y)=\mathop{ \min } \limits_{g\in G(y)}\left\{\int_{\mathbb{D}}{{\psi(g)\mathrm{d}\mathbb{F}}}
\right\}.
\end{equation}
where $\psi$ is the Legendre transformation of $\phi$, i.e., $\psi(y)=\mathop{\sup}\limits_{x>0}[\phi(x)-xy]$, $y>0$. 
Denote by $\hat{g}(y)$ the solution to Problem~(\ref{D2}).

Next, we study the duality relation between the value functions and solutions of the two dual problems (\ref{p3}) and (\ref{D2}) to derive the solution to Problem~(\ref{p0}).
Based on Proposition \ref{set} and dual method, we investigate the distorted conjugation between the two value functions $u(\cdot)$ and $v(\cdot)$ and obtain the following distorted duality theorem.
\vskip 10pt
{{
\begin{theorem}[\bf{distorted duality theorem}]\label{duality}
Assume that the ambiguity attitude $\phi$ is strictly concave and satisfies the Inada condition ($\phi'(0)=\infty$, $\phi'(\infty)=0$) and that $u(x)<\infty$ for some $x>0$. Suppose also that Assumption \ref{Ass} holds  and $h(x)$ is invertible. Then

(i) $u(x)<\infty$, for all $x > 0$, and there exists $y_0 > 0$ such that $v(y)$ is finitely
valued for $y>y_0$. The value functions $u$ and $v$ are distorted conjugate, i.e.,
\begin{equation}\label{dual}
\begin{cases}
u(x)=\mathop {\inf}\limits_{y>0}[v(y)+h(x)y],\\
v(y)=\mathop {\sup}\limits_{x>0}[u(x)-h(x)y].
\end{cases}
\end{equation}

The function $u$ is continuously differentiable on $(0,\infty)$ and the function $v$ is strictly convex on $\{v<\infty\}$.
The functions $u'$ and $-v'$ satisfy
\begin{equation}\nonumber
u'(0)=\lim\limits_{x\rightarrow 0}u'(x)=\infty,~\text{   }~v'(\infty)=\lim\limits_{y\rightarrow\infty}v'(y)=0.
\end{equation}

(ii) If $v(y)<\infty$, then the solution $\hat{g}(y)\in\mathbf{G}(y)$ to Problem~(\ref{D2}) exists and is unique.
\end{theorem}
\begin{proof}
	See Appendix \ref{proof of duality}.
\end{proof}


Theorem \ref{duality} shows that the value functions $u$ and $v$ are distorted conjugate and the solution to Problem~(\ref{D2}) for $y\in\{v<\infty\}$ exists and is unique. Let $z=h(x)$, and we have
\begin{equation}\label{dualnormal}
\begin{cases}
u\left(h^{-1}(z)\right)=\mathop {\inf}\limits_{y>0}[v(y)+zy],\\
v(y)=\mathop {\sup}\limits_{h^{-1}(z)>0}\left[u\left(h^{-1}(z)\right)-{z}y\right],
\end{cases}
\end{equation}
which shows that $u\!\circ\!h^{-1}$ and $v$ are conjugate. 

In order to derive the solution to Problem~(\ref{p3}), we assume also that $AE(\phi)<1$, which makes $v$ well-defined and continuously differentiable on $(0,\infty)$. Then the existence and uniqueness of the solution to Problem~(\ref{p3}) and the dual relation of the solutions to Problems (\ref{p3}) and (\ref{D2}) are given in the following theorem.

\begin{theorem}\label{solution}
In addition to the assumptions in Theorem \ref{duality}, suppose that $AE(\phi)\!\triangleq\limsup\limits_{x\rightarrow+\infty}\frac{x\phi'(x)}{\phi(x)}<1$ (see \cite{25}).
Then in addition to the assertions of Theorem \ref{duality}, we have:

(i) $v(y)<\infty$, for all $y > 0$. The value functions $u$ and $v$ are continuously differentiable on $(0,\infty)$ and the functions $u'$ and $-v'$ are strictly decreasing and satisfy
\begin{equation}\nonumber
u'(\infty)=\lim\limits_{x\rightarrow \infty}u'(x)=0,~\text{   }~v'(0)=\lim\limits_{y\rightarrow0}v'(y)=\infty.
\end{equation}

(ii) The optimal solution $\hat{f}(x)\in\mathbf{B}(x)-\mathbf{D}$ to Problem~(\ref{p3}) exists and is unique. If $\hat{g}(y)\in\mathbf{G}(y)$ is the optimal solution to Problem~(\ref{D2}), where $y=\frac{u'\left(x\right)}{h'(x)}$, we have the dual relation
\begin{equation}\label{dualsol}
\hat{f}(x)=\left(\phi'\right)^{-1}\!\left(\hat{g}(y)\right),~\text{   }~\hat{g}(y)=\phi'\left(\hat{f}(x)\right).
\end{equation}
Moreover,
\begin{equation}
\int_{\mathbb{D}}\hat{f}(x)\hat{g}(y)\mathrm{d}\mathbb{F}=h(x)y.
\end{equation}

(iii) We have the following relations between $u'$, $\hat{f}(x)$ and $v'$, $\hat{g}(y)$, respectively,
\begin{equation}\nonumber
u'(x)=\int_{\mathbb{D}}\frac{h'(x)}{h(x)}\hat{f}\phi'(\hat{f})\mathrm{d}\mathbb{F},~\text{   }~v'(y)=\int_{\mathbb{D}}\frac{\hat{g}\psi'(\hat{g})}{y}\mathrm{d}\mathbb{F}.
\end{equation}
\end{theorem}
\begin{proof}
	See Appendix \ref{proof of solution}.
\end{proof}

In Theorem \ref{solution}, Eq.~(\ref{dualsol}) is similar to the results in the classical EUT problem, while the only difference is that $y=\frac{u'(x)}{h'(x)}$ instead of $y=u'(x)$ because the value functions $u$ and $v$ become distorted conjugate. Moreover, as $\hat{g}(y)\in\mathbf{G}(y)$, there exists $\lambda\in\mathbf{S}$ such that $\hat{g}(y)\leq \frac{\lambda}{\rho(\lambda)}$. Then
\begin{equation}\nonumber
\int_{\mathbb{D}}\frac{\hat{g}(y)}{y}\mathrm{d}\mathbb{F}\leq \int_{\mathbb{D}}\frac{y\lambda}{\rho(\lambda)}\mathrm{d}\mathbb{F}
=\frac{1}{\rho(\lambda)}
\leq 1,
\end{equation}
and the equality holds if and only if $\hat{g}(y)=\frac{\lambda}{\rho(\lambda)}$ and $\hat{X}^\lambda=S^0$, which only holds for some specific cases. In most cases, the equality does not hold, and we have
\begin{equation}\nonumber
\int_{\mathbb{D}}\frac{\hat{g}(y)}{y}\mathrm{d}\mathbb{F}< 1,
\end{equation}
which is similar to the result $\mathbf{E}[\hat{Y}_T]<y$ of the classical EUT problem in an incomplete market.
}}


\vskip 4pt
\vskip 3pt
\vskip 4pt

\vskip 4pt
In this subsection, we have transformed the original problem~(\ref{p0}) to an equivalent one on the efficient frontier. 
{{
Then we introduce Assumption \ref{Ass} and derive  the bipolar relation, the distorted duality theorem and the dual relation of the solutions to the equivalent problem~(\ref{p3}) and the dual problem~(\ref{D2}). 
Finally,}}
we summarize the above results {{in the whole section}} and present the formulation of the solution to Problem~(\ref{p0}) as follows. 
\vskip 10pt
\begin{theorem}\label{TT5}
Under the assumptions of {{Theorem \ref{solution}}}, the optimal terminal  wealth of Problem~(\ref{p0}) is
\begin{equation}{\label{XThati}}
{{\hat{X}^{\hat{\lambda}}_T}} = I\left(\hat{Y}^{\hat{\lambda}}_T\right){e^{rT}},
\end{equation}
where $\hat{\lambda}$ satisfies
\begin{equation}{\label{equa1i}}
\hat{\lambda}(\mu)=\frac{\phi'\left({{\mathbf{b}(\mu,\hat{X}^{\hat{\lambda}})}}\right)}{<\phi'\left({{\mathbf{b}(\cdot,\hat{X}^{\hat{\lambda}})}}\right),\mathbb{I}_{\mathbb{D}}>},\\
\end{equation}
{{
and $I(\cdot)=(U')^{-1}(\cdot)$, $\hat{Y}^{\hat{\lambda}}_T$ is the solution to Problem~(\ref{infY}).

}}
The optimal investment strategy of Problem~(\ref{p0}) is
{{
\begin{equation*}
\hat{\pi}^{\hat{\lambda}}_t=\frac{\hat{f}^{\hat{\lambda}}_t}{\hat{f}^{\hat{\lambda}}_t-1} \hat{X}^{\hat{\lambda}}_t,
\end{equation*}
}}
where $\hat{f}$ is determined by
{{
\begin{equation}\nonumber
\mathbf{E}^{{}^{\hat{\lambda}} Q'}\left[\frac{e^{rT}}{\hat{X}^{\hat{\lambda}}_T}\mid\mathcal{F}^S_t\right]= \frac{1}{x} + \int_0^t \hat{f}^{\hat{\lambda}}_u\mathrm{d}\left(\frac{S^1_u e^{ru}}{\hat{X}^{\hat{\lambda}}_u}\right),
\end{equation}
and ${}^{\hat{\lambda}} Q'$ is defined in Eq.~(\ref{mart}).
}}
Moreover, the value function is
{{
\begin{equation}\nonumber
u(x)=\mathop {\max }\limits_{\pi\in\mathcal{V}[0,T]}\left\{\Phi({{X_T}})\right\} = \int_{\mathbb{D}}\phi\left(\mathbf{b}(\cdot,\hat{X}^{\hat{\lambda}})\right)\mathrm{d}\mathbb{F}.
\end{equation}
}}
\end{theorem}
\begin{proof}
	See Appendix \ref{proof of TT5}.
\end{proof}
\vskip 4pt
{{
Different from the classical EUT problem, in the pre-commitment KMM model, the DM considers the nonlinear weighted average of the expected utilities, where there are two composite expectations of nonlinear functions of the variables. Theorem \ref{TT5} shows that the optimal wealth process $\hat{X}^{\hat{\lambda}}$ is the combination of the solutions to Problems (\ref{pr2}) and (\ref{p3}), which are determined by their dual problems (\ref{infY}) and (\ref{D2}), respectively.
Because the solution to Problem~(\ref{infY}) is fairly complicated in an incomplete market, we also present the results in Theorem \ref{TT5} in a complete market in the following corollary.

\begin{corollary}\label{complete}
Under the assumptions of Theorem \ref{TT5}, and assume also that the financial market is complete, i.e., $\mathcal{M}=\{Q\}$, then the optimal terminal wealth of Problem~(\ref{p3}) is
\begin{equation}{\label{XThat}}
\hat{X}_T^{\hat{\lambda}} = I\left(\kappa(\hat{\lambda})\frac{\mathrm{d}Q}{\mathrm{d}{}^{\hat{\lambda}}Q}\right)e^{rT},
\end{equation}
where $\hat{\lambda}$ and $\kappa(\hat{\lambda})$ satisfy
\begin{equation}{\label{equa1}}
\left\{
\begin{split}
&\hat{\lambda}(\mu)=\frac{\phi'\left({{\mathbf{b}(\mu,\hat{X}^{\hat{\lambda}})}}\right)}{<\phi'\left({{\mathbf{b}(\cdot,\hat{X}^{\hat{\lambda}})}}\right),\mathbb{I}_{\mathbb{D}}>},\\
&\mathbf{E}^{Q}\left[I\left(\kappa(\hat{\lambda})\frac{\mathrm{d}Q}{\mathrm{d}{}^{\hat{\lambda}}Q}\right)\right]=x.
\end{split}
\right.
\end{equation}
The optimal investment strategy of Problem~(\ref{p3}) is
\begin{equation*}
\hat{\pi}(t)=\hat{f}(t),
\end{equation*}
where $\hat{f}$ is determined by
\begin{equation}\nonumber
\mathbf{E}^{Q}\left[e^{-rT}\hat{X}^{\hat{\lambda}}_T|\mathcal{F}^S_t\right]
=x + \int_0^t \hat{f}_u\frac{\mathrm{d}\left[e^{-rt}S^1_t\right]}{S^1_t}.
\end{equation}
Moreover, the value function is
\begin{equation}\nonumber
u(x)=\mathop {\max }\limits_{\pi\in\mathcal{V}[0,T]}\left\{\Phi(X_T)\right\} = \int_{\mathbb{D}}\phi\left(\mathbf{b}(\cdot,\hat{X}^{\hat{\lambda}})\right)\mathrm{d}\mathbb{F}.
\end{equation}
\end{corollary}
}}

\setcounter{equation}{0}
\section{{{\bf Black-Scholes Financial Market}}}
In this section, we  introduce specific settings of the financial market and the DM's attitudes towards risk and ambiguity. For simplicity, to characterize ambiguity attitude, we choose $\phi(x)=x^\gamma/\gamma$, $\gamma<1$  {if $U$ takes value in $[0,+\infty)$, and $\phi(x)=-(-x)^\gamma/\gamma$, $\gamma<1$ if $U$ takes value in $[-a,0]$ for some $a>0$}. For the utility function $U$, we consider three different cases: CARA function, HARA function, and CRRA function. We check  Assumption~\ref{Ass}  and apply  Theorem \ref{TT5} to derive the optimal strategies of Problem~(\ref{p0}) for these three utility functions.

\cite{Blanchard93} and \cite{BT18} both show that the first moments of stock returns are hard to be estimated accurately. As such, the stocks' yields are often ambiguous. We consider a Black-Scholes financial market  consisting of a risk-free asset ${{S^0}}$ and an ambiguity asset ${{S^1}}$ as follows:
\begin{equation*}
\left\{
\begin{array}{ll}
\mathrm{d}{{S^0_t}}=r{{S^0_t}}\mathrm{d}t,&{{S^0_0=s^0_0}},\\
\mathrm{d}{{S^1_t}}={{S^1_t}}(\mu_* \mathrm{d}t+\sigma\mathrm{d}{{W^*_t}}),~&{{S^1_0=s^1_0}},
\end{array}
\right.
\end{equation*}
where $r>0$ is the risk-free interest rate, the real yield ${\mu_*}$ is {{only $\mathcal{F}$-measurable but not $S^1$-measurable, which means that the DM can}} not estimate it accurately. The volatility $\sigma>0$ is a known constant, and $W^*$ is a standard $\mathbf{P}$-Brownian motion. We assume that the DM is ambiguous about the real yield ${\mu_*}$ of asset $S_1$. The DM has some subjective  estimates or priors towards $\mu_*$. ${\mu_*}$ takes values in $\mathbb{R}$ and has SOD $F(\cdot)$ whose probability density function w.r.t. Lebesgue measure is denoted by $p(\cdot)$.
\vskip 4pt
For any fixed realization $\mu\in\mathbb{R}$, define ${{{}^*\eta_t^\mu}}=\exp\left\{-\frac{\mu_*-\mu}{\sigma}{{W^*_t}}-\frac{(\mu_*-\mu)^2}{2\sigma^2}t\right\}$, and $\frac{\mathrm{d}Q^\mu}{\mathrm{d}\mathbf{P}}|_{\mathcal{F}_t}={{{}^*\eta_t^\mu}}$. Then the process $\left\{{{W^{Q^{\mu}}_t}}\triangleq {{W^*_t}}+\frac{\mu_*-\mu}{\sigma}t,~0\leq t\leq T\right\}$ is a {{$\mathcal{F}^S_t$-}}Brownian motion under $Q^\mu$, {{which is observable for the DM}}. 
As such, under $Q^\mu$, the dynamic process of the ambiguity asset $S_1$ has the form:
\begin{equation*}
\mathrm{d}{{S^1_t}}={{S^1_t}}\left(\mu\mathrm{d}t+\sigma\mathrm{d}{{W^{Q^\mu}_t}}\right),~{{S^1_0=s^1_0}}.
\end{equation*}

{{
As such, we successfully transform the ambiguity of the yield $\mu_*$ to the ambiguity on $\mathcal{P}=\{Q^\mu, \mu \in \mathbb{R}\}$ equivalently, which is a specific case of the model in Section 2. 
Observing that $Q^\mu\approx\mathbf{P}$, $\forall Q^\mu\in\mathcal{P}$, we know that all the measures in $\mathcal{P}$ are equivalent and any $Q^\mu\in\mathcal{P}$ is an domination of $\mathcal{P}$. Assume that
the DM has a reference point described by $\mu_0$ and $\hat{Q}=Q^{\mu_0}$ (denote $W^{Q^{\mu_0}}$ by $\hat{W}$ for simplicity).}}
In statistics,  $\mu_0\in\mathbb{R}$ represents the statistical estimation of the past stock yield, and $Q^{\mu_0}$ represents the subjective estimation of the real probability $\mathbf{P}$ of the  financial market.
{{
For $\mu\in\mathbb{R}$, we know
\begin{equation*}
{{\eta^\mu_t}}=\exp\left\{-\nu_\mu \hat{W}_t-\frac{1}{2}\nu_\mu^2t\right\},
\end{equation*}
where
\begin{equation*}
\nu_\mu=\frac{\mu_0-\mu}{\sigma}.
\end{equation*}
}}

Let $\pi=\{{{\pi_t}},{0\leq t\leq T}\}$ be the money invested in ${{S^1}}$. The admissible strategy $\pi$ is  required to satisfy that ${{\pi_t}}$ is $\mathcal{F}_t^S-$measurable, and  $\int_0^T\!{{\pi^2_t}}\mathrm{d}t\!<\!+\infty,\!~a.s.\!~{{\hat{Q}}}$. 
Then the DM's wealth process $X=\left\{{{X_t}},0\leq t\leq T \right\}$  with strategy $\pi$ under $Q^\mu$ is given by
 \begin{equation}\label{xi}
 \left\{
 \begin{array}{ll}
 \mathrm{d}{{X_t}}&=[r{{X_t}}+{{\pi_t}}(\mu-r)]\mathrm{d}t+{{\pi_t}}\sigma\mathrm{d}{{\hat{W}_t}},\\
 {{X_0}}&=x.
 \end{array}
 \right.
 \end{equation}
Define
\begin{equation*}
{{\eta_t}}=\exp\{-\nu {{\hat{W}_t}}-\frac{1}{2}\nu^2t\},\\
\end{equation*}
where
\begin{equation*}
\nu=\frac{\mu_0-r}{\sigma},\\
\end{equation*}
{{and the probability measure $Q$ satisfies $\frac{\mathrm{d}Q}{\mathrm{d}\hat{Q}}|_{\mathcal{F}^S_T}=\eta_T$}}. Then, under 
$Q$, the process $X$ satisfies
 \begin{equation*}
 \begin{array}{ll}
 \mathrm{d}{{X_t}}&=r{{X_t}}\mathrm{d}t+\sigma{{\pi_t}}\mathrm{d}{{W^{Q}_t}},
 \end{array}
 \end{equation*}
where $\{{{W^Q_t}},0\leq t\leq T\}$ is a Brownian motion under $Q$. Besides,  we have ${{W^Q_t}}-\nu t={{W^{Q^\mu}_t}}-\nu_\mu t={{\hat{W}_t}}$ for $\mu\in\mathbb{R}$. Also, $\{e^{-rt}{{X_t}},0\leq t\leq T\}$ is a {{local}} martingale under probability measure $Q$, i.e., $\{e^{-rt}{{\eta_t}}{{X_t}},0\leq t\leq T\}$ is a {{local}} martingale under probability measure ${{\hat{Q}}}$. Moreover, it is observed that the financial is complete, i.e., $\mathcal{M}=\{Q\}$. 
Based on the martingale representation theorem,
there exists a stochastic process $\{{{\hat{f}_t}}, 0\leq t\leq T  \}$ such that
\begin{equation}\label{X2}
\begin{array}{ll}
e^{-rt}{{\eta_t}}{{\hat{X}_t}} &= \mathbf{E}^{{{\hat{Q}}}}\left[e^{-rT}I\left(\frac{\hat{\kappa}{{\eta_T}}}{{{{}^\lambda \eta_T}}}
\right){{\eta_T}}|\mathcal{F}^S_t\right]\\
&=x + \int_0^t {{\hat{f}_u}}\mathrm{d}{{\hat{W}_u}}.
\end{array}
\end{equation}
Using Eq.~(\ref{xi}), we also have
\begin{equation}\label{X3}
\begin{array}{ll}
\mathrm{d}\left[e^{-rt}{{\eta_t}}{{\hat{X}_t}}\right]=e^{-rt}{{\eta_t}}\left(-\nu{{\hat{X}_t}}+ \sigma{{\hat{\pi}_t}}\right)\mathrm{d}{{\hat{W}_t}}.
\end{array}
\end{equation}
Comparing the diffusion coefficients in Eqs.~(\ref{X2}) and (\ref{X3}), we obtain the optimal strategy:
\begin{equation*}
{{\hat{\pi}_t}}=\frac{e^{rt}{{\eta^{-1}_t\hat{f}_t}}+\nu\hat{X}_t}{\sigma}.
\end{equation*}
Now we present the optimal strategies for three different utility functions of $U(\cdot)$. In order to apply Corollary \ref{complete}, we need to prove that Assumption~\ref{Ass} holds. Then, we can obtain the explicit forms of the investment strategy in the case that the SOD is normal distribution: $N(\mu_0,\sigma_\mu^2)$.

\subsection{\bf CARA Utility}
\subsubsection{{\bf Validity of condition~(\ref{asp})}}
We first prove that Assumption~\ref{Ass} holds, which only depends on the utility function $U(\cdot)$ and is not influenced by the specific model of the risky asset and the ambiguity attitudes $\phi(\cdot)$. We will see that in these three cases,  Assumption~\ref{Ass} is mainly related with the form of the utility function $U(\cdot)$. The utility function is $U(x)=-\frac{1}{\alpha}e^{-\alpha x}, \alpha>0$, and the inverse of $U'$ is $I(x)=-\frac{1}{\alpha}\log (x)$. Then, based on Eqs.~(\ref{XThat})-(\ref{equa1}), the terminal wealth ${{\hat{X}^\lambda_T}}$ corresponding to  function $\lambda$ is given by
\[
{{\hat{X}^\lambda_T}}=-\frac{1}{\alpha}\log\left(\frac{\kappa{{\eta_T}}}{{{{}^\lambda\eta_T}}}\right),\]
with budget constraint
\[
xe^{rT}=-\frac{1}{\alpha}\mathbf{E}^Q\left[\log(\kappa)+\log({{\eta_T}})-\log({{{}^\lambda\eta_T}})
\right].
\]
Solving the last equation, we have
\[\kappa=\exp\left\{-\alpha x e^{rT}-\mathbf{E}^{Q}\left[\log(\frac{{{\eta_T}}}{<\lambda(\cdot),{{\eta^\cdot_T}}>})\right]\right\}.\]
In addition,
\[U\!\circ\! I(x)=-\frac{1}{\alpha}x.\]
As such, 
\begin{equation}\nonumber
\begin{split}
<\mathbf{E}^{Q^{\bigcdot}}[U({{\hat{X}^\lambda_T}})],\lambda(\cdot)>&=~<-\frac{1}{\alpha}\mathbf{E}^{{{\hat{Q}}}}\left[\frac{\kappa{{\eta_T}}{{\eta^\cdot_T}}}{{{{}^\lambda\eta_T}}}\right],\lambda(\cdot)>\\
&=-\frac{\kappa}{\alpha}\mathbf{E}^{{{\hat{Q}}}}\left[{{\eta_T}}\frac{{{{}^\lambda\eta_T}}}{{{{}^\lambda\eta_T}}}\right]\\
&=-\frac{1}{\alpha}\exp\left\{-\alpha x e^{rT}-\mathbf{E}^{Q}\left[\log(\frac{{{\eta_T}}}{{{{}^\lambda\eta_T}}})\right]\right\}\\
&=h(x)\rho(\lambda),\\
\end{split}
\end{equation}
where $h(x)=U(xe^{rT})=-\frac{1}{\alpha}e^{-\alpha x e^{rT}}$ and $\rho(\lambda)=\exp\left\{-\mathbf{E}^{Q}\left[\log(\frac{{{\eta_T}}}{{{{}^\lambda\eta_T}}})\right]\right\}$. Thus the condition in Eq.~(\ref{asp}) is satisfied.

We observe that not only the condition of Eq.~(\ref{asp}) is satisfied, but also the function $h(x)$ is just equal to  $U(xe^{rT})$. As such, $h(x)$ is invertible and differentiable, which satisfies the assumptions in  Theorem \ref{TT5}. Thus, it allows us to apply Corollary \ref{complete} to present the optimal strategy of Problem~(\ref{p0}).

\subsubsection{{\bf Optimal strategy of Problem~(\ref{p0})}}
Before deriving the strategy, we first simplify the formulations and conditions in Corollary \ref{complete}.  
Using the relationship
\begin{equation*}
\frac{\phi'(\mathbf{E}^{Q^\mu}[U(\hat{X}_T)])}{<\phi'(\mathbf{E}^{Q^{\bigcdot}}[U(\hat{X}_T)]),\mathbb{I}_{\mathbb{D}}>}=\hat{\lambda}(\mu),
\end{equation*}
we eliminate $\hat{\lambda}$ in Eq.~(\ref{XThat}) in 
Corollary \ref{complete} and {{have the fact that 
\begin{equation}\nonumber
\begin{split}
\kappa \eta_T&=I^{-1}(\hat{X}_T)<\hat{\lambda}(\cdot),\eta^\cdot_T>\\
&=U'(\hat{X}_T)\frac{1}{\tilde{\kappa}}\int_{\mathbb{R}}\phi'\left(\mathbf{E}^{Q^\mu}[U(\hat{X}_T)]\right)\eta^\mu_T \mathrm{d}F(\mu),
\end{split}
\end{equation}
where $\tilde{\kappa}=<\phi'(\mathbf{E}^{Q^{\cdot}}[U(\hat{X})]),\mathbb{I}_{\mathbb{D}}>$. Regarding $\kappa \tilde{\kappa}$ as a new constant $\kappa$ and denoting by $p(\mu)$ the probability density function of $F(\mu)$,}}
we obtain the conditions only related with the optimal terminal wealth ${{\hat{X}^{\hat{\lambda}}_T}}$ of Problem~(\ref{p0}) {(denoted as $\hat{X}_T$ for simplicity in this section)}:
\begin{equation}\label{cond}
\left\{
\begin{split}
&U'(\hat{X}_T)\int_{\mathbb{R}}\phi'\left(\mathbf{E}^{Q^\mu}[U(\hat{X}_T)]\right)\eta^\mu_T p(\mu)\mathrm{d}\mu=\kappa {{\eta_T}},\\
&\mathbf{E}^Q[{{\hat{X}_T}}]=\mathbf{E}^{{{\hat{Q}}}}[{{\hat{X}_T}}{{\eta_T}}]=x_0e^{rT}.
\end{split}
\right.
\end{equation}
\vskip 4pt
In statistics, normal distribution is often adopted in parameter estimation. For all examples in this section, we assume that ambiguous yield $\mu$ follows the normal distribution 
$N(\mu_0, \sigma_\mu^2 )$, which is equivalent to  $\nu_\mu\sim N(0, \frac{\sigma_\mu^2}{\sigma^2} )$. In order to simplify the formulations of the results, we simply denote $\sigma_0^2=\frac{\sigma^2}{\sigma_\mu^2T}$.

\vskip 4pt
{Noting that $U(x)=-\frac{1}{\alpha} e^{-\alpha x}\in[-\frac{1}{\alpha},0]$ for $x>0$, we choose $\phi(x)=-\frac{(-x)^\gamma}{\gamma}$, $\gamma<1$, and use Eq. (\ref{cond}) to derive the optimal solution.}
In order to obtain explicit form of the optimal strategy, we guess that the optimal terminal wealth has the form ${{\hat{X}_T}}=\frac{1}{\alpha}(\frac{p}{2T}{{\hat{W}_T^2}}+q{{\hat{W}_T}}+c)$. As such,
\begin{equation}\nonumber
\begin{split}
\mathbf{E}^{Q^\mu}[U(\hat{X}_T)]&=\mathbf{E}^{{{\hat{Q}}}}\left[-\frac{1}{\alpha}e^{-\alpha\hat{X}_T}{{\eta^\mu_T}}\right]\\
&=\mathbf{E}^{{{\hat{Q}}}}\left[-\frac{1}{\alpha}\exp\left\{-\frac{p}{2T}{{\hat{W}_T}}-q{{\hat{W}_T}}-\nu_\mu {{\hat{W}_T}}-\frac{1}{2}\nu_\mu^2T-c\right\}\right]\\
&=-\frac{1}{\alpha\sqrt{1+p}}\exp\left\{-\frac{pT}{2(p+1)}\nu_\mu^2+\frac{qT}{p+1}\nu_\mu+\frac{q^2T}{2(p+1)}-c\right\}.
\end{split}
\end{equation}
As such,
\begin{equation}\label{eq41}
\begin{split}
&\int_{\mathbb{R}}\phi'\left(\mathbf{E}^{Q^\mu}[U(\hat{X}_T)]\right)\eta^\mu_T p(\mu)d\mu\\
\propto&\int_{\mathbb{R}}e^{-(\gamma-1)\frac{pT}{2(p+1)}\nu_\mu^2+(\gamma-1)\frac{qT}{p+1}\nu_\mu-\frac{T}{2}\nu_\mu^2-{{\hat{W}_T}}\nu_\mu}e^{-\frac{\sigma_0^2T}{2}\nu_\mu^2}\mathrm{d}\nu_\mu\\
\propto&\exp\left\{\frac{1}{2T(\frac{1+\gamma p}{1+p}+\sigma_0^2)}\left({{\hat{W}_T^2}}-\frac{2(\gamma-1)Tq}{1+p}{{\hat{W}_T}}\right)\right\}.
\end{split}
\end{equation}

Substituting Eq.~(\ref{eq41}) into Eq.~(\ref{cond}), we obtain the following equations:
\begin{equation}\label{cara}
\left\{
\begin{split}
&(\gamma+\sigma_0^2)p^2+\sigma_0^2 p-1=0,\\
&\frac{1+\gamma p}{1+p}q=\nu,\\
&\frac{1}{2}(\nu^2T+1)p-\nu Tq+c=\alpha xe^{rT}.
\end{split}
\right.
\end{equation}
By the solution $(p,q,c)$ of Eq.~(\ref{cara}) with $p>-1$, we obtain the optimal terminal wealth ${{\hat{X}_T}}=\frac{1}{\alpha}(\frac{p}{2T}{{\hat{W}_T^2}}+q{{\hat{W}_T}}+c)$, where
\begin{equation}\label{result1}
\left\{
\begin{split}
&p=\frac{\sqrt{\sigma_0^4+4\sigma_0^2+4\gamma}-\sigma_0^2}{2(\sigma_0^2+\gamma)},\\
&q=\frac{1+p}{1+\gamma p}\nu,\\
&c=\alpha xe^{rT}-\frac{p}{2}+\frac{1+\frac{1}{2}p-\frac{1}{2}\gamma p^2}{1+\gamma p}\nu^2 T.
\end{split}
\right.
\end{equation}
\vskip 4pt
We see from Eq.~(\ref{result1}) that $p>0$, which means that the DM prepares for bad cases of the future market (i.e., ${{\hat{W}_T}}<0$). In the traditional models ignoring ambiguity, the variance of the yield $\sigma_\mu^2=0$, i.e., $\sigma_0^2\rightarrow \infty$, as such, $p\rightarrow 0$ and $q=\frac{1+p}{1+\gamma p}\nu\rightarrow \nu$. Then, the optimal terminal wealth reduces to the traditional form ${{X_T}}=\frac{\nu}{\alpha} {{\hat{W}_T}}+c'$ in non-ambiguity case. In non-ambiguity  case, $p=0$, and ${{X_T}}$ is a linear function of ${{\hat{W}_T}}$. However, wee see that, in ambiguity case, $0<p<1$, and ${{X_T}}$ is a quadratic function of ${{\hat{W}_T}}$. As such, when the DM considers ambiguity of the yield, he/she tends to give up some benefits of normal situations with small bias  to ensure the benefits of extreme situations with large bias.
\vskip 4pt
In order to obtain the optimal strategy, based on martingale method, we first calculate ${{\hat{X}_T}}$ as follows:
\begin{equation}\nonumber
\begin{split}
{{\hat{X}_t}}=~&\mathbf{E}^{Q}[e^{-r(T-t)}{{\hat{X}_T}}|\mathcal{F}^S_t]\\
=~&\mathbf{E}^{{{\hat{Q}}}}\left[\frac{1}{\alpha}e^{-r(T-t)}\left(\frac{p}{2T}{{\hat{W}_T^2}}+q{{\hat{W}_T}}+c\right)e^{-\nu ({{\hat{W}_T}}-{{\hat{W}_t}})-\frac{1}{2}\nu^2(T-t)}|\mathcal{F}^S_t\right]\\
=~&\frac{1}{\alpha}e^{-r(T-t)}\left(\frac{p}{2T}{{\hat{W}_t^2}}+q{{\hat{W}_t}}+c\right)-\frac{\nu(T-t)p}{\alpha T}e^{r(T-t)}{{\hat{W}_t}}.\\
\end{split}
\end{equation}
{Notice that $X_t$ is a quaduatic function of $W_t$, as such, the wealth process becomes nonnegative if $x$ is large enough, which is because $U'(x)=e^{-\alpha x}$ is not a surjection from $(0,+\infty)$ to $(0,+\infty)$.}
As such, $e^{-rT}{{\hat{X}_T}}{{\eta_T}}$ is replicated by
\begin{equation}\nonumber
\begin{split}
e^{-rT}{{\hat{X}_T}}{{\eta_T}}=x+\int_0^T&\left[\frac{\nu^2(T-t)p}{\alpha T}{{\hat{W}_t}}-\frac{\nu(\nu^2(T-t)+3)(T-t) p}{2\alpha T}\right.\\
&\left.+\frac{\nu^2(T-t)q}{\alpha}\right]e^{-\nu {{\hat{W}_t}}-\frac{1}{2}\nu^2t}\mathrm{d}{{\hat{W}_t}}.
\end{split}
\end{equation}
Thus, the optimal strategy is given by
\begin{equation*}
\begin{split}
\hat{\pi}_t=\frac{1}{\sigma}e^{rt}\left(\frac{\nu^2(T-t)p}{\alpha T}{{\hat{W}_t}}-\frac{\nu(\nu^2(T-t)+3)(T-t) p}{2\alpha T}+\frac{\nu^2(T-t)q}{\alpha}\right)+\frac{\nu}{\sigma}{{\hat{X}_t}},
\end{split}
\end{equation*}
where $p$ and $q$ are given by Eq.~(\ref{result1}). In the traditional expected utility optimization problem with CARA utility function, the optimal investment strategy is only a function of time. However, we see that when considering ambiguity, the DM should adjust the strategy relying on the states ${{\hat{W}_t
}}$ and ${{\hat{X}_T}}$ at time $t$.

\subsection{\bf CRRA Utility}
\subsubsection{{\bf Validity of condition~(\ref{asp})}}
In this case, we suppose that the utility function is $U(x)=\frac{1}{\beta}x^\beta, \beta<1$, and the inverse of $U'$ is $I(x)=x^{\frac{1}{\beta-1}}$. Based on Eqs.~(\ref{XThat})-(\ref{equa1}), the terminal wealth ${{\hat{X}^\lambda_T}}$ corresponding to a fixed weight function $\lambda$ is
\[
{{\hat{X}^\lambda_T}}=\left(\frac{\kappa{{\eta_T}}}{{{{}^\lambda\eta_T}}}\right)^{\frac{1}{\beta-1}}\]
with budget constraint
\[xe^{rT}=\kappa^{\frac{1}{\beta-1}}
\mathbf{E}^Q\left[\left(\frac{{{\eta_T}}}{{{{}^\lambda\eta_T}}}\right)^{\frac{1}{\beta-1}}\right].
\]
It follows that
\[\kappa=(xe^{rT})^{\beta-1}\left(\mathbf{E}^Q
\left[\left(\frac{{{\eta_T}}}{{{{}^\lambda\eta_T}}}\right)^{\frac{1}{\beta-1}}
\right]\right)^{-(\beta-1)}.\]
Besides,
\[U\!\circ\! I(x)=\frac{1}{\beta}x^{\frac{\beta}{\beta-1}}.\]
As such, 
\begin{equation}\nonumber
\begin{split}
<\mathbf{E}^{Q^{\bigcdot}}[U({{\hat{X}^\lambda_T}})],\lambda(\cdot)>&=~<\frac{1}{\beta}\mathbf{E}^{{{\hat{Q}}}}\left[\left(\frac{\kappa{{\eta_T}}}{{{{}^\lambda\eta_T}}}\right)^{\frac{\beta}{\beta-1}}{{\eta^\cdot_T}}\right],\lambda(\cdot)>\\
&=\frac{1}{\beta}\kappa^{\frac{\beta}{\beta-1}}\mathbf{E}^{{{\hat{Q}}}}\left[{{\eta_T}}\left(\frac{{{\eta_T}}}{{{{}^\lambda\eta_T}}}\right)^{\frac{1}{\beta-1}}\right]\\
&=\frac{1}{\beta}(xe^{rT})^\beta\left(\mathbf{E}^{Q}\left[\left(\frac{{{\eta_T}}}{{{{}^\lambda\eta_T}}}\right)^{\frac{1}{\beta-1}}\right]\right)^{1-\beta}\\
&=h(x)\rho(\lambda),\\
\end{split}
\end{equation}
where $h(x)=U(xe^{rT})=\frac{1}{\beta}(xe^{rT})^\beta$ and $\rho(\lambda)=\left(\mathbf{E}^{Q}\left[\left(\frac{{{\eta_T}}}{{{{}^\lambda\eta_T}}}\right)^{\frac{1}{\beta-1}}\right]\right)^{1-\beta}$. Thus the condition in Eq.~(\ref{asp}) follows.
\vskip 4pt
The result is similar to the former case that not only the condition  of Eq.~(\ref{asp})  is satisfied, but also the function $h(x)$ is just equal to  $U(xe^{rT})$. As such, $h(x)$ is also invertible and differentiable, which satisfies the assumptions in  Theorem \ref{TT5}. Therefore, it allows us to use the formulation in Theorem \ref{TT5} to  present  the optimal strategy of Problem~(\ref{p0}).
\vskip 4pt
\subsubsection{\bf Optimal strategy of Problem~(\ref{p0})}
Based on Eq.~(\ref{cond}), we guess that the optimal terminal wealth has the form : \[{{\hat{X}_T}}=\exp{\left\{\frac{1}{\beta}\left(\frac{p}{2T}{{\hat{W}_T^2}}+q{{\hat{W}_T}}+c\right)\right\}}.\]
As such, 
\begin{equation}\label{EXT}
\begin{split}
\mathbf{E}^{Q^\mu}[U(\hat{X}_T)]&=\mathbf{E}^{{{\hat{Q}}}}\left[\frac{1}{\beta}\hat{X}_T^\beta{{\eta^\mu_T}}\right]\\
&=\mathbf{E}^{{{\hat{Q}}}}\left[\frac{1}{\beta}\exp\left\{\frac{p}{2T}{{\hat{W}_T^2}}+q{{\hat{W}_T}}-\nu_\mu {{\hat{W}_T}}-\frac{1}{2}\nu_\mu^2T+c\right\}\right]\\
&=\frac{1}{\beta\sqrt{1-p}}\exp\left\{\frac{pT}{2(1-p)}
\nu_\mu^2-\frac{qT}{1-p}\nu_\mu+\frac{q^2T}{2(1-p)}+c\right\}.
\end{split}
\end{equation}
As such,
\begin{equation}\label{eq42}
\begin{split}
&\int_{\mathbb{R}}\phi'\left(\mathbf{E}^{Q^\mu}[U(\hat{X}_T)]\right)\eta^\mu_T p(\mu)d\mu\\
\propto&\int_{\mathbb{R}}e^{(\gamma-1)\frac{pT}{2(1-p)}\nu_\mu^2-(\gamma-1)\frac{qT}{1-p}\nu_\mu-\frac{T}{2}\nu_\mu^2-{{\hat{W}_T}}\nu_\mu}e^{-\frac{\sigma_0^2T}{2}\nu_\mu^2}\mathrm{d}\nu_\mu\\
\propto&\exp\left\{\frac{1}{2T(\frac{1-\gamma p}{1-p}+\sigma_0^2)}\left({{\hat{W}_T^2}}+\frac{2(\gamma-1)Tq}{1-p}{{\hat{W}_T}}\right)\right\}.
\end{split}
\end{equation}
Substituting Eq.~(\ref{eq42}) into Eq.~(\ref{cond}), we obtain the following equations:
\begin{equation}\nonumber
\left\{
\begin{split}
&(\gamma+\sigma_0^2)p^2-(\sigma_0^2+\frac{1}{1-\beta}) p+\frac{\beta}{1-\beta}=0,\\
&\frac{(1-\beta)(1-\gamma p)}{\beta(1-p)}q=\nu,\\
&\exp\left\{\frac{(\nu-\frac{q}{\beta})^2}{2(1-\frac{p}{\beta})}T-\frac{1}{2}\nu^2T+\frac{c}{\beta}\right\}=xe^{rT}.
\end{split}
\right.
\end{equation}
The solution$(p,q,c)$ of the last equation can be obtained if $p<1$. As such,  the optimal terminal wealth is \[{{\hat{X}_T}}=\exp{\left\{\frac{1}{\beta}\left(\frac{p}{2T}{{\hat{W}_T^2}}+q{{\hat{W}_T}}+c\right)\right\}},\]
where
\begin{equation}\label{result2}
\left\{
\begin{split}
&p=\frac{\frac{1}{1-\beta}+\sigma_0^2-\sqrt{\sigma_0^4+\frac{2-4\beta}{1-\beta}\sigma_0^2+\frac{1}{(1-\beta)^2}-\frac{4\beta}{1-\beta}\gamma}}{2(\sigma_0^2+\gamma)},\\
&q=\frac{\beta(1-p)}{(1-\beta)(1-\gamma p)}\nu,\\
&c=\beta\left[\log(x)+rT+\left(\nu^2-\frac{(\nu-\frac{q}{\beta})^2}{(1-\frac{p}{\beta})}\right)\frac{ T}{2}\right].
\end{split}
\right.
\end{equation}

It is similar to the CARA case with $p>0$, which means that the DM prepares for bad cases of the future market (i.e., ${{\hat{W}_T}}<0$). In the traditional models without ambiguity, the variance of the yield $\sigma_\mu^2=0$, i.e., $\sigma_0^2\rightarrow \infty$, as such, $p\rightarrow 0$ and $q=\frac{\beta(1-p)}{(1-\beta)(1-\gamma p)}\nu\rightarrow \frac{\beta}{1-\beta}\nu$. Then, the optimal terminal wealth reduces to the traditional form ${{X_T}}=\exp{\left\{\frac{\nu}{\beta-1} {{\hat{W}_T}}+c'\right\}}$ in non-ambiguity case. In non-ambiguity case, $p=0$, and $\log {{X_T}}$ is a linear function of ${{\hat{W}_T}}$. However, in ambiguity case, $0<p<1$, and $\log {{X_T}}$ is a quadratic function of ${{\hat{W}_T}}$.
\vskip 5pt
To obtain the optimal strategy, we first calculate ${{\hat{X}_T}}$ as follows:
\begin{equation}\nonumber
\begin{aligned}
{{\hat{X}_t}}=~&\mathbf{E}^{Q}\left[e^{-r(T-t)}{{\hat{X}_T}}|\mathcal{F}_t\right]\\
=~&\mathbf{E}^{{{\hat{Q}}}}\left[e^{-r(T-t)}e^{\frac{1}{\beta}\left(\frac{p}{2T}{{\hat{W}_T^2}}+q{{\hat{W}_T}}+c\right)}e^{-\nu ({{\hat{W}_T}}-{{\hat{W}_t}})-\frac{1}{2}\nu^2(T-t)}|\mathcal{F}_t\right]\\
=~&e^{\frac{1}{\beta}(\frac{p}{2T}{{\hat{W}_t^2}}+q{{\hat{W}_t}}+c)}e^{-r(T-t)}\sqrt{\frac{\beta T}{\beta T-(T-t)p}}e^{\frac{\beta T (T-t)}{2(\beta T-(T-t)p)}(\frac{p}{\beta T}{{\hat{W}_t}}+\frac{q}{\beta}-\nu)^2  -\frac{1}{2}\nu^2(T-t) }.\\
\end{aligned}
\end{equation}
Then $e^{-rT}{{\hat{X}_T}}{{\eta_T}}$ is replicated by
\begin{equation}\nonumber
\begin{split}
e^{-rT}{{\hat{X}_T}}{{\eta_T}}=x+\int_0^T\frac{p{{\hat{W}_t}}+Tq-\beta T\nu}{\beta T-(T-t)p}e^{-rt}{{\eta_t}}{{\hat{X}_t}}\mathrm{d}{{\hat{W}_t}}.
\end{split}
\end{equation}
Finally, the optimal strategy is
\begin{equation}\label{picrra}
\begin{split}
\hat{\pi}_t=\frac{1}{\sigma}\frac{1}{\beta T-(T-t)p}\left[p{{\hat{W}_t}}+Tq-(T-t)p\nu\right]{{\hat{X}_t}},
\end{split}
\end{equation}
where $p$ and $q$ are given by Eq.~(\ref{result2}). In the standard EUT problem with CRRA utility, the optimal strategy is proportional to the wealth. Considering ambiguity, we see from the above equation that the {}{observable}} market state ${{\hat{W}_t}}$ exists in the optimal strategy, which is quite different.
\subsection{\bf HARA Utility}
\subsubsection{{\bf Validity of condition~(\ref{asp})}}
In order to simplify the formulation, we choose the HARA utility function of the form $U(x)=\frac{1}{\beta}(x+a)^\beta$. The inverse of $U'$ is $I(x)=x^{\frac{1}{\beta-1}}-a$. Then, based on Eqs.~(\ref{XThat})-(\ref{equa1}), the terminal wealth ${{\hat{X}^\lambda_T}}$ corresponding to a fixed weight function $\lambda$ is
\begin{equation}\nonumber
\begin{split}
{{\hat{X}^\lambda_T}}&=\left(\frac{\kappa{{\eta_T}}}{{{{}^\lambda\eta_T}}}\right)^{\frac{1}{\beta-1}}-a,
\end{split}
\end{equation}
with budget constraint
\begin{equation}\nonumber
	\begin{split}
xe^{rT}&=\kappa^{\frac{1}{\beta-1}}\mathbf{E}^Q\left[\left(\frac{{{\eta_T}}}{{{{}^\lambda\eta_T}}}\right)^{\frac{1}{\beta-1}}\right]-a.
\end{split}
\end{equation}
It follows that
\[\kappa=(xe^{rT}+a)^{\beta-1}
\left(\mathbf{E}^Q\left[\left(\frac{{{\eta_T}}}{{{{}^\lambda\eta_T}}}\right)^{\frac{1}{\beta-1}}
\right]\right)^{-(\beta-1)},\]
and
\[U\!\circ\! I(x)=\frac{1}{\beta}x^{\frac{\beta}{\beta-1}}.\]
As such, 
\begin{equation}\nonumber
\begin{split}
<\mathbf{E}^{Q^{\bigcdot}}[U({{\hat{X}^\lambda_T}})],\lambda(\cdot)>&=~<\frac{1}{\beta}\mathbf{E}^{{{\hat{Q}}}}\left[\left(\frac{\kappa{{\eta_T}}}{{{{}^\lambda\eta_T}}}\right)^{\frac{\beta}{\beta-1}}{{\eta^\cdot_T}}\right],\lambda(\cdot)>\\
&=\frac{1}{\beta}\kappa^{\frac{\beta}{\beta-1}}\mathbf{E}^{{{\hat{Q}}}}\left[{{\eta_T}}\left(\frac{{{\eta_T}}}{{{{}^\lambda\eta_T}}}\right)^{\frac{1}{\beta-1}}\right]\\
&=\frac{1}{\beta}(xe^{rT}+a)^\beta\left(\mathbf{E}^{Q}\left[\left(\frac{{{\eta_T}}}{{{{}^\lambda\eta_T}}}\right)^{\frac{1}{\beta-1}}\right]\right)^{1-\beta}\\
&=h(x)\rho(\lambda),
\end{split}
\end{equation}
where $h(x)=U(xe^{rT})=\frac{1}{\beta}(xe^{rT}+a)^\beta$ and $\rho(\lambda)=\left(\mathbf{E}^{Q}\left[\left(\frac{{{\eta_T}}}{{{{}^\lambda\eta_T}}}\right)^{\frac{1}{\beta-1}}\right]\right)^{1-\beta}$.
\vskip 4pt
Thus the condition in Eq.~(\ref{asp}) follows. We observe that the functions $\rho(\lambda)$ of HARA case and CRRA case are the same because HARA utility is equal to CRRA utility with a
translation of the terminal wealth. 
As such, any translation of the independent variable 
on the utility function $U(\cdot)$ will not impact the condition in Assumption~\ref{Ass}. And it is obvious that not only the condition  of Eq.~(\ref{asp}) is  satisfied, but also the function $h(x)$ is still equal to the utility function $U(xe^{rT})$. Thus, $h(x)$ is also invertible and differentiable, which satisfies the assumptions in Theorem \ref{TT5}. Also, it allows us to use the formulation in Theorem \ref{TT5} to present the optimal strategy of Problem~(\ref{p0})in next subsection.

\subsubsection{{\bf Optimal strategy of Problem~(\ref{p0})}}
We know that HARA utility just equals to CRRA with a translation of the independent variable.  
As such, the optimal terminal wealth should also equal to the CRRA case with a translation. We suppose  \[{{\hat{X}_T}}=e^{\frac{1}{\beta}\left(\frac{p}{2T}{{\hat{W}_T^2}}+q{{\hat{W}_T}}+c\right)}-a.\]
For $p,q,c$, using Theorem \ref{TT5}, we have
\begin{equation}\label{result3}
\left\{
\begin{split}
&p=\frac{\frac{1}{1-\beta}+\sigma_0^2-\sqrt{\sigma_0^4+\frac{2-4\beta}{1-\beta}\sigma_0^2+\frac{1}{(1-\beta)^2}-\frac{4\beta}{1-\beta}\gamma}}{2(\sigma_0^2+\gamma)},\\
&q=\frac{\beta(1-p)}{(1-\beta)(1-\gamma p)}\nu,\\
&c=\beta\left[\log(x+a)+rT+\left(\nu^2-\frac{(\nu-\frac{q}{\beta})^2}{(1-\frac{p}{\beta})}\right)\frac{ T}{2}\right].
\end{split}
\right.
\end{equation}
The results in Eq.~(\ref{result3}) are similar to Eq.~(\ref{result2}) except for $c$. As such, the optimal strategy has  a similar form as in CRRA case, which is
\begin{equation}\nonumber
\begin{split}
\hat{\pi}_t=\frac{1}{\sigma}\frac{1}{\beta T-(T-t)p}\left[p{{\hat{W}_t}}+Tq-(T-t)p\nu\right]{{\hat{X}_t}},
\end{split}
\end{equation}
where $p$ and $q$ are given in Eq.~(\ref{result3}).

\vskip 10pt

Summarizing the former three different cases, we find that $h(x)=U(xe^{rT})$ holds for all these cases. Comparing with the classical EUT problem, we find that the smooth ambiguity problem with the value $xe^{rT}$ of the budget constraint, utility function $U$ and ambiguity attitude $\phi$, can be seen as another expected utility problem such as $\int_{\mathbb{D}}\phi(\cdot)\mathrm{d}\mathbb{F}$ on the range $\mathbb{D}$ of the ambiguous yield $\mu$ with the value $U(xe^{rT})$ of the distorted budget constraint. 
This means that{{, in a complete market,}} the DM {{with above three kinds of utility functions}} only needs to consider the one-fold expected utility optimization problem with the given ambiguity  attitude and a distorted budget constraint. The distortion function $h(x)$ equals to $U(xe^{rT})$, which can be understood as the utility at time $T$ by only investing in the risk-free asset. Therefore, when the DM makes decisions, it is enough for him/her to consider the objective $\int_{\mathbb{D}}\phi(\cdot)\mathrm{d}\mathbb{F}$ and the value $U(xe^{rT})$ of the distorted budget constraint, then derive the strategy corresponding to solution of $\int_{\mathbb{D}}\phi(\cdot)\mathrm{d}\mathbb{F}$ problem.
\vskip 5pt
In addition, we find that all the optimal wealths in the former three cases with ambiguity contain the term ${{\hat{W}_T^2}}$ with positive coefficient, while those without ambiguity only contain ${{\hat{W}_T}}$ but {do not contain} ${{\hat{W}_T^2}}$. This phenomenon illustrates that, 
considering model uncertainty, the DM with ambiguity aversion will be cautious about the possible 
extreme cases in ambiguous cases. The DM tends to give up some benefits of normal situations with small bias to ensure the benefits of extreme situations far with large bias.

\setcounter{equation}{0}
\section{{ {\bf Sensitivity Analysis}}}
In this section we present some numerical examples to show the effects of different parameters on the efficient frontier and optimal investment strategy. We show the influences on efficient frontier and optimal investment strategy separately. If we consider the efficient frontier in continuous case, there exists a map  from the unit sphere in $\mathbf{L}^2(\mathbb{R})$ to the efficient frontier, which is an infinite dimensional manifold and unable to be presented in a figure. In the first subsection, we choose the discrete case to show the effects of different parameters on the efficient frontier. In the second subsection, we compare the DMs with and without ambiguity. In the third subsection, we show the optimal investment strategy in the CRRA case.  Unless otherwise  stated, the basic parameters are given by $\mu_0=0.1$, $ r=0.05$, $\sigma=0.2$, $x=1$, $ \sigma_0=2$, $\beta=\frac{1}{3}$, $\gamma=-0.5$, $T=4$.

\subsection{\bf Efficient frontier}
In this subsection, the model of financial market is the same as in Section 4, and there are only two subjective priors, i.e., $\mathbb{D}=\{\mu_1,~\mu_2\}$. As such, the efficient frontier $\mathbf{B}(x)$ becomes an arc, which is shown in Fig.~\ref{fig:frontier}. The end points of the arcs in Fig.~\ref{fig:frontier} represent the utilities of the cases that the DM optimizes the utility under $Q^{\mu_1}$ or $Q^{\mu_2}$. In the discrete case, we suppose that  $U(x)=\frac{1}{\beta}x^\beta, \phi(x)=\frac{1}{\gamma}x^\gamma$, $\mu_1=0.15$, $\mu_2=0.09$, $p(\mu_1)=\frac{2}{3}$ and $p(\mu_2)=\frac{1}{3}$. Based on Theorem \ref{TT5}, the numerical illustrations for the efficient frontier are presented in Fig.~\ref{fig:frontier}.

\begin{figure}[h]
  \centering
  \includegraphics[totalheight=6cm]{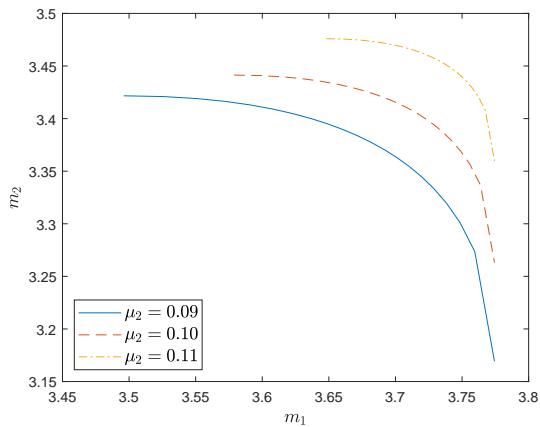}
\caption{{  Efficient frontier.}}
  \label{fig:frontier}
\end{figure}

Fig.~\ref{fig:frontier} shows the efficient frontier $\mathbf{B}\!=\!\!\left\{\!(m_1, m_2)\!\!=\!\!\left(\mathbf{E}^{Q^{\mu_1}}\!\!\left[U\!({{X_T}})\right]\!, \mathbf{E}^{Q^{\mu_2}}\!\!\left[U\!({{X_T}})\right]\right)\!\right\}$ for different $\mu_2$. As mentioned in Section 3, 
using variation of {{$\lambda=(\lambda_1, \lambda_2)$ in $\mathbf{S}=\{(\lambda_1, \lambda_2)|\lambda_i\geq 0, \lambda_1 p(\mu_1)+\lambda_2 p(\mu_2)=1\}$}} and Monte Carlo simulation, we get the efficient frontier $\mathbf{B}$. For fixed $\mu_2$, we find that there is a trade-off between $m_1$ and $m_2$. $m_2$ becomes smaller as $m_1$ increases. Because if $(m^1_1, m^1_2)$ lies on the efficient frontier $\mathbf{B}$, then for any other $(m^2_1, m^2_2)\in\mathbf{B}$ with $m^2_1>m^1_1$, we must have $m^2_2<m^1_2$. When the DM is more concerned with the utility under some prior, it is natural to expect that the utility under other prior decreases as the DM is ambiguity averse.

\subsection{Comparison between the cases with/without (consideration of) ambiguity}
Next, we suppose that the SOD is Gaussian. The results with Gaussian SOD are presented in the last section with different utility functions.
{{In the case with (consideration of) ambiguity, the DM considers $\mu\in\mathbb{R}$ with Gaussian SOD and $\mathcal{P}=\left\{Q^\mu,\mu\in\mathbb{R}\right\}$, while, in the case without (consideration of) ambiguity, the DM only consider the point $\mu_0$ and $\mathcal{P}=\left\{Q^{\mu_0}\right\}$, which means that the market is still ambiguous but the DM just ignores it.}}

First we compare the optimal terminal wealths with and without ambiguity. Based on Eq.~(\ref{result2}) of the optimal terminal wealth with ambiguity, the coefficient of {{${\hat{W}_T}^2$}} is $\frac{p}{2T\beta}=0.0363$, the coefficient of {{$\hat{W}_T$}} is $\frac{q}{\beta}=0.3230$; while, in the optimal terminal wealth without ambiguity, the coefficient of {{${\hat{W}_T}^2$}} is $0$ and the coefficient of {{$\hat{W}_T$}} is $\frac{1}{1-\beta}\nu=0.3750$. As such, different from the DM ignoring  ambiguity, the DM considering ambiguity lowers the coefficient of {{$\hat{W}_T$}} and adds a small coefficient of {{${\hat{W}_T}^2$}}, which results in some losses when ${{\hat{W}_T}}$ is positive while not too large and more benefit when {{$|\hat{W}_T|$}} is large.
\vskip 4pt
Then, based on Eq.~(\ref{EXT}), the expected utility of the case with ambiguity under different probability measure $Q^\mu$ is given by
\begin{equation}\nonumber
\begin{split}
\mathbf{E}^{Q^\mu}[U(\hat{X}_T)]&=\frac{1}{\beta\sqrt{1-p}}\exp\left\{\frac{pT}{2(1-p)}
\nu_\mu^2-\frac{qT}{1-p}\nu_\mu+\frac{q^2T}{2(1-p)}+c\right\}\\
&=\frac{1}{\beta\sqrt{1-p}}\exp\left\{\frac{pT}{2(1-p)}\frac{(\mu-\mu_0)^2}{\sigma^2}+\frac{qT}{1-p}\frac{\mu-\mu_0}{\sigma}+\frac{q^2T}{2(1-p)}+c\right\};
\end{split}
\end{equation}
the expected utility of the case without ambiguity under different measure $Q^\mu$ is given by
\begin{equation}\nonumber
\begin{split}
\mathbf{E}^{Q^\mu}[U(\hat{X}_T)]&=\frac{1}{\beta}\exp\left\{-\frac{\beta}{1-\beta}\nu T\nu_\mu+\beta\log x+\beta rT+\frac{\beta}{2(1-\beta)}\nu^2 T\right\}\\
&=\frac{1}{\beta}\exp\left\{\frac{\beta\nu T}{1-\beta}\frac{\mu-\mu_0}{\sigma}+\beta\log x+\beta rT+\frac{\beta}{2(1-\beta)}\nu^2 T\right\}.
\end{split}
\end{equation}

The comparison of the expected utilities of the two cases is shown in Fig. \ref{fig:utilitycompare}.
\begin{figure}[h]
  \centering
  \includegraphics[totalheight=6cm]{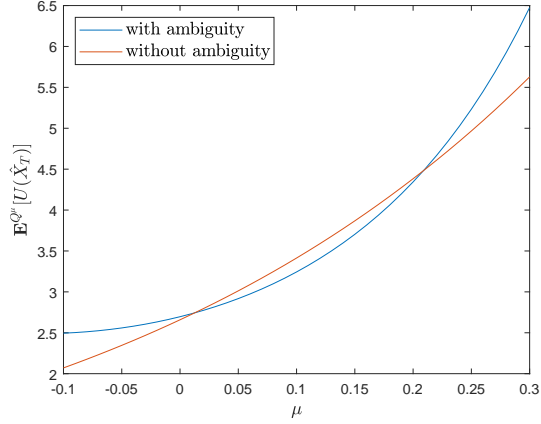}
\caption{{  Comparison of $\mathbf{E}^{Q^\mu}[U(\hat{X}_T)]$ in two cases.}}
  \label{fig:utilitycompare}
\end{figure}
In Fig.~\ref{fig:utilitycompare}, compared with the DM without ambiguity, the expected utility of the DM decreases a little when the yield is close to the point estimation $\mu_0=0.1$ but increases when the yield is far from the point estimation $\mu_0=0.1$, which means that the DM with ambiguity attitude tends to give up some benefits of normal situations with small bias to ensure the benefits of extreme situations with large bias.
\vskip 4pt
We also compare the value functions with and without ambiguity in Fig.~\ref{fig:loss}.
\begin{figure}[h]
  \centering
  \includegraphics[totalheight=6cm]{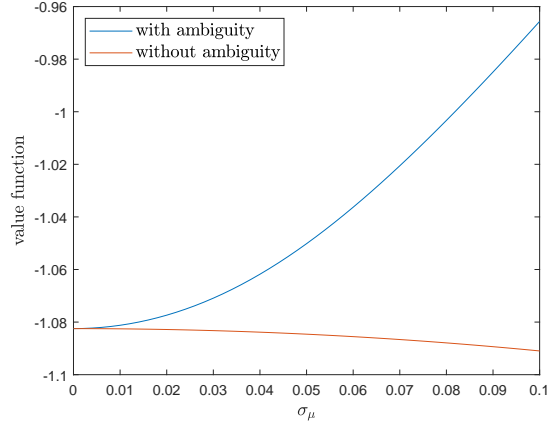}
\caption{{  Comparison of value function in two cases.}}
  \label{fig:loss}
\end{figure}
In Fig.~\ref{fig:loss}, the value functions are the same in the two cases if there is no ambiguity in the financial market ($\sigma_\mu=0$). However, when the uncertainty over the financial market increases ($\sigma_\mu$ grows), the value function of the DM considering ambiguity increases, while the value function of the DM ignoring ambiguity decreases, which means that the DM is faced with utility loss when ignoring ambiguity.

\vskip 4pt
Finally, we consider the comparison of the feedback functions of the optimal strategy. Based on Eq.~(\ref{picrra}), {
{
We find that $\hat{\pi}_0$ only depends on the initial value $x$ but no longer depends on $\hat{W}$ because $\hat{W}_0=0$ is a constant, which is different from $\hat{\pi}_t$ depending on $\hat{W}_t$ when $t\in(0,T]$. Here we show the feedback form of $\hat{\pi}$ when $t\in(0,T]$ in Fig.~\ref{fig:feedback} and discuss $\hat{\pi}_0$ in detail in the next subsection.
}
} %
\begin{figure}[h]
  \centering
  \includegraphics[totalheight=6cm]{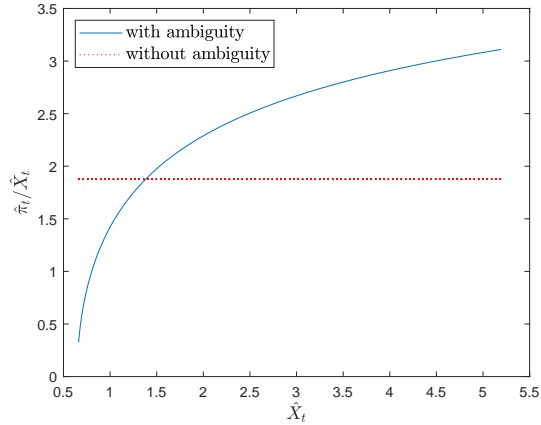}
\caption{{  Comparison of feedback functions in two cases.}}
  \label{fig:feedback}
\end{figure}

Fig.~\ref{fig:feedback} shows that the optimal investment proportion is a constant for 
the DM {{without ambiguity}}. However, the 
DM {{considering ambiguity with ambiguity aversion}} invests more (less) in the risky asset when the financial market performs well (worse).

\subsection{\bf Optimal investment strategy at time $t=0$}
In this subsection, we study the performance of optimal investment strategy $\hat{\pi}$ at time $t=0$. We first investigate the optimal investment strategy {{$\hat{\pi}_t$}} at time $t=0$, and the optimal strategy in the CRRA case is
\begin{equation}\nonumber
\begin{split}
{{\hat{\pi}_0}}=\frac{1}{\sigma}\cdot\frac{q-p\nu}{\beta-p}x,
\end{split}
\end{equation}
where
\begin{equation}\nonumber
\left\{
\begin{split}
&p=\frac{\frac{1}{1-\beta}+\sigma_0^2-\sqrt{\sigma_0^4+\frac{2-4\beta}{1-\beta}\sigma_0^2+\frac{1}{(1-\beta)^2}-\frac{4\beta}{1-\beta}\gamma}}{2(\sigma_0^2+\gamma)},\\
&q=\frac{\beta(1-p)}{(1-\beta)(1-\gamma p)}\nu,
\end{split}
\right.
\end{equation}
which becomes a linear function of the initial value $x$. Then we show the effects of different parameters on $\hat{\pi}_0$ in the following figures.

\begin{figure}[h]
  \centering
  \includegraphics[totalheight=6cm]{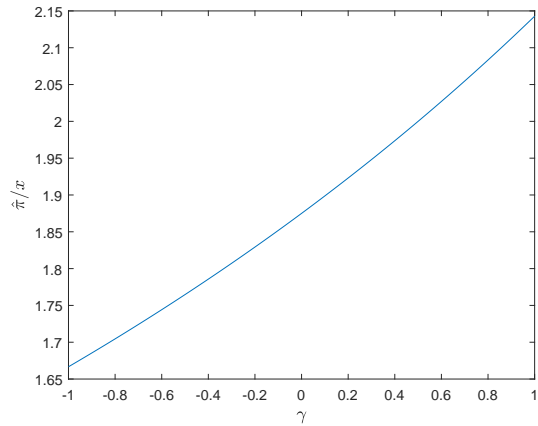}
\caption{{  Effect of $\gamma$ on $\hat{\pi}$.}}
  \label{fig:gamma}
\end{figure}

Fig.~\ref{fig:gamma} illustrates that {{$\hat{\pi}_0$}} increases with $\gamma$. As a larger $\gamma$ means less ambiguity aversion, DM with less ambiguity aversion surely invests more in ambiguity asset. Especially when $\gamma=1$, the ambiguity attitude becomes $\phi(x)=x$, which means that DM is no longer ambiguity averse and tends to be ambiguity neutral.

Next, we study the effects of $\beta$ and variance of $\mu$ (i.e.,  $\sigma_\mu^2=\frac{\sigma^2}{\sigma_0^2T}$) on the optimal strategy at time $t=0$ in both ambiguity averse and ambiguity neutral cases, and illustrate differences between the two cases.
\begin{figure}[h]
  \centering
  \includegraphics[totalheight=6cm]{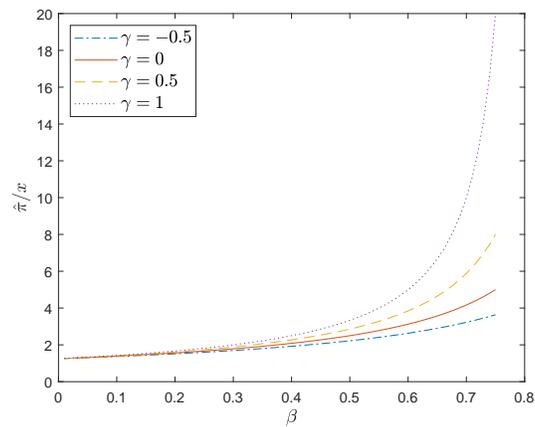}
  \caption{ {  Effect of $\beta$ on $\hat{\pi}$.}}
  \label{fig:beta}
\end{figure}

Fig.~\ref{fig:beta} shows that {{$\hat{\pi}_0$}} increases when $\beta$ becomes larger in both ambiguity averse and ambiguity neutral cases. As mentioned before, bigger $\beta$ means less risk aversion, DM will increase his/her position in the ambiguity asset. But it is also shown in Fig.~\ref{fig:beta} that {{$\hat{\pi}_0$}} increases more rapidly when $\gamma$ decreases.  In the ambiguity neutral case, the speed of increase is much more rapid than that in ambiguity averse case.

\begin{figure}[h]
  \centering
  \includegraphics[totalheight=6cm]{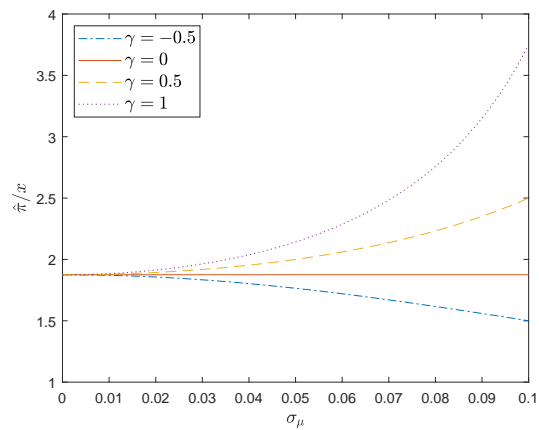}
  \caption{ {  Effect of $\sigma_\mu$ on $\hat{\pi}$.}}
  \label{fig:sigma}
\end{figure}

Fig.~\ref{fig:sigma} shows that {{$\hat{\pi}_0$}} increases as $\sigma_\mu$ becomes larger in ambiguity neutral case and low ambiguity averse case while {{$\hat{\pi}_0$}} decreases with $\sigma_\mu$ in high ambiguity averse case. The critical point of high ambiguity aversion and low ambiguity aversion is $\gamma=0$, i.e., $\phi(x)=\log(x)$. In this example, the ambiguous yield $\mu \sim N(\mu_0,\frac{\sigma^2}{\sigma_0^2T})$. As such, larger $\sigma_\mu$ means larger uncertainty of the yield, DM with high ambiguity aversion will invest less in the ambiguous asset. However, ambiguity neutral DMs and even DMs with low ambiguity aversion will invest more in the ambiguity asset. Besides, the DMs with ambiguity attitude $\phi(x)=\log(x)$ always invest a proportion of  $\frac{\nu x}{\sigma(1-\beta)}$ independent of $\sigma_\mu$. It is also shown in Fig.~\ref{fig:sigma} that {{$\hat{\pi}_0$}} in the ambiguity neutral case varies much more dramatically than that in the ambiguity averse case.
\vskip 5pt
\section{{ {\bf Concluding Remarks}}}
In this paper, we present and solve the pre-commitment KMM problem proposed by 
\cite{7} {{for dominated priors in an incomplete financial market}} in a continuous-time framework. The existence and uniqueness of the solution of the KMM problem are analysed by convex analysis, distorted Legendre transformation and {{distorted}} duality theorem. We find that: (i) {{The two-fold expected utility optimization problem is equivalent to the combination of two kinds of one-fold expected utility optimization problems (Problem~(\ref{pr2}) on the first-order space and Problem~(\ref{p3}) on the second-order space). 
(ii) Under Assumption \ref{Ass}, we introduce the distorted Legendre transformation and derive the bipolar relation in Theorem \ref{bipolarThm} and the distorted duality theorem in Theorem \ref{duality} for Problem~(\ref{p3}). Under the further assumption $AE(\phi)<1$, the existence and uniqueness of the solution to Problem~(\ref{p3}) and the dual relation of the solutions to Problem~(\ref{p3}) and the dual problem (\ref{D2}) are obtained in Theorem \ref{solution}.}} {Finally, the solution of the KMM problem is given in Theorem \ref{TT5}.} 
(iii)  The KMM problem with initial value $x$ is equivalent to the optimization problem $\int_{\mathbb{D}}\phi(\mu)\mathrm{d}F(\mu)$ on the ambiguous range $\mathbb{D}$ with the value $U(xe^{rT})$ of the distorted budget constraint {{for some common utility functions in a complete market.}}  
(iv) In the case of Gaussian SOD in a Black-Scholes financial market, the optimal terminal wealth relies on a quadratic function
of the market state, which leads to the fact that the DM with ambiguity aversion would rather give up some benefits of normal situations with small bias to ensure the benefits of extreme situations with large bias. {{Moreover, the feedback function of the optimal investment strategy is linear and non-linear of the wealth process for $t=0$ and $t\in(0,T]$, respectively.}}
\vskip 4pt
In our work, we consider {{dominated priors to formulate the problem clearly. If the non-dominated priors are taken into consideration, it is difficult to construct the efficient frontier and separate the first-order and second-order problems, which is left for further studies.}}

\vskip 15pt
{\bf Acknowledgements.}
The authors acknowledge the support from the National Natural Science Foundation of China (Grant No.12271290, No.11901574, No.11871036). The authors thank Dr. Fengyi Yuan and  the members of the group of Actuarial Science and Mathematical Finance at the Department of Mathematical Sciences, Tsinghua University for their feedbacks and useful conversations.

\appendix
\section{\bf{ Proof of Theorem \ref{TT1}. } }\label{proof of TT1}
\begin{proof}[{\bf{ Proof of Theorem \ref{TT1} } }]
{\bf\textit{  If side.}} Suppose that {{$\hat{X}\in\mathfrak{X}(x)$}} 
satisfies
\begin{equation}\label{T1}
<{{\mathbf{b}(\cdot,\hat{X})}}-{{\mathbf{b}(\cdot,X)}},\lambda(\cdot)>\geq 0.
\end{equation}
for any {{$X\in\mathfrak{X}(x)$.}} 
If ${{\mathbf{b}(\cdot,\hat{X})}}\notin\mathbf{B}(x)$, then there exists a point ${{\mathbf{b}(\cdot,\tilde{X})}}\in \mathbf{H}(x)$ such that \[{{\mathbf{b}(\cdot,\tilde{X})}}-{{\mathbf{b}(\cdot,\hat{X})}}\in {\mathbf{L}^0_+(\mathbb{D}, \mathcal{B}(\mathbb{D}), \mathbb{F})}\setminus\{0\}.\]
It follows that
\begin{equation}\label{T2}
<{{\mathbf{b}(\cdot,\tilde{X})}}-{{\mathbf{b}(\cdot,\hat{X})}},\lambda(\cdot)>\geq 0.
\end{equation}
Combining Eq.~(\ref{T1}) and Eq.~(\ref{T2}), as $\lambda\in{\mathbf{L}^0_+(\mathbb{D}, \mathcal{B}(\mathbb{D}), \mathbb{F})}\setminus\{0\}$ {{and the support set of $\lambda$ is denoted by $\Lambda\subset\mathbb{D}$}}, we have 
\[{{\mathbf{b}(\cdot,\tilde{X})}}={{\mathbf{b}(\cdot,\hat{X})}}\ \mbox{ on } \Lambda.\]
As such, ${{\hat{X}}}$ and ${{\tilde{X}}}$ both maximize $<{{\mathbf{b}(\cdot,X)}},\lambda(\cdot)>$.

Suppose that $ {{\hat{Q}}}({{\hat{X}_T}}\neq {{\tilde{X}_T}})> 0$, then by the concavity of utility function $U(\cdot)$,
\begin{equation}\nonumber
\begin{split}
<{{\mathbf{b}\left(\cdot,\frac{\hat{X}+\tilde{X}}{2}\right)}},\lambda(\cdot)>
>\frac{1}{2}<{{\mathbf{b}(\cdot,\hat{X})}},\lambda(\cdot)>
+\frac{1}{2}<{{\mathbf{b}(\cdot,\tilde{X})}},\lambda(\cdot)>.
\end{split}
\end{equation}
As such,
\[{{\hat{Q}}}\left({{\hat{X}_T}}\neq {{\tilde{X}_T}}\right)= 0,\]
which implies
{{
\[Q^\mu\left(\hat{X}_T\neq\tilde{X}_T\right)=0,~\forall \mu\in\mathbb{D},\]
}}
and
\[{{\mathbf{b}(\cdot,\tilde{X})}}-{{\mathbf{b}(\cdot,\hat{X})}}=0 \notin \mathbf{L}^0_+(\mathbb{D}, \mathcal{B}(\mathbb{D}), \mathbb{F})\setminus\{0\},\]
which is in contradiction with the condition that ${{\hat{X}}}$ attains the maximum. Thus,
 ${{\mathbf{b}(\cdot,\hat{X})}}\in \mathbf{B}(x)$.
\vskip 5pt
{\bf\textit{ Only if side.}}  We first prove that $\mathbf{B}(x)-\mathbf{D} = \{\mathbf{b}-\mathbf{d}{\in\mathbf{L}^0_+(\mathbb{D}, \mathcal{B}(\mathbb{D}), \mathbb{F})}|\mathbf{b}\in \mathbf{B}(x), \mathbf{d} \in {\mathbf{L}^0_+(\mathbb{D}, \mathcal{B}(\mathbb{D}), \mathbb{F})}\}$ is a convex set. Assume $\mathbf{b}_i={{\mathbf{b}(\cdot,X^i)}}\in \mathbf{B}(x), \mathbf{d}_i\in{\mathbf{L}^0_+(\mathbb{D}, \mathcal{B}(\mathbb{D}), \mathbb{F})}, i=1,2$. $\forall$ $0<\theta<1$ and $\mu\in\mathbb{D}$,
\begin{equation*}
\begin{array}{llll}
\theta  \mathbf{b}_1(\mu) +(1-\theta)\mathbf{b}_2(\mu)
&=&\theta{{\mathbf{b}(\mu,X^1)}}+(1-\theta){{\mathbf{b}(\mu,X^2)}}\\
&=&\theta\mathbf{E}^{Q^\mu}\left[U\left({{X^1_T}}\right)\right]+(1-\theta)\mathbf{E}^{Q^\mu}\left[U\left({{X^2_T}}\right)\right]\\
&\leq& \mathbf{E}^{Q^\mu}\left[U\left(\theta {{X^1_T}}+(1-\theta) {{X^2_T}} \right)\right].
\end{array}
\end{equation*}
We claim that $ \mathbf{E}^{Q^{\bigcdot}}\left[U\left(\theta {{X^1_T}}+(1-\theta) {{X^2_T}}\right)\right]$ can be expressed by $\mathbf{b}_3-\mathbf{d}_3$ with $\mathbf{b}_3\in \mathbf{B}(x), \mathbf{d}_3 \in {\mathbf{L}^0_+(\mathbb{D}, \mathcal{B}(\mathbb{D}), \mathbb{F})}$.

In fact, if  $\mathbf{E}^{Q^{\bigcdot}}\left[U\left(\theta {{X^1_T}}+(1-\theta) {{X^2_T}} \right)\right]\in \mathbf{B}(x)$, let $\mathbf{d}_3=0$. Otherwise, based on Definition \ref{D1}, there exist $\mathbf{b}_3\in\mathbf{B}(x)$ and $\mathbf{d}_3\in{\mathbf{L}^0_+(\mathbb{D}, \mathcal{B}(\mathbb{D}), \mathbb{F})}\setminus\{0\}$ such that
 \[\mathbf{E}^{Q^{\bigcdot}}\left[U\left(\theta {{X^1_T}}+(1-\theta) {{X^2_T}} \right)\right]=\mathbf{b}_3-\mathbf{d}_3.\]
As such, there exists a $\mathbf{d}_4\in{\mathbf{L}^0_+(\mathbb{D}, \mathcal{B}(\mathbb{D}), \mathbb{F})}$ such that
\begin{equation}\nonumber
\theta  \mathbf{b}_1 +(1-\theta)\mathbf{b}_2=\mathbf{b}_3-\mathbf{d}_4.
\end{equation}
Then
\begin{equation}\nonumber
\begin{array}{lll}
\theta  (\mathbf{b}_1-\mathbf{d}_1) +(1-\theta)(\mathbf{b}_2-\mathbf{d}_2)
=\mathbf{b}_3-\left[\theta \mathbf{d}_1+(1-\theta)\mathbf{d}_2+\mathbf{d}_4\right]\in \mathbf{B}(x)-\mathbf{D},
\end{array}
\end{equation}
i.e., $\mathbf{B}(x)-\mathbf{D}$ is a convex set.

{{
Then we prove that $\mathbf{B}(x)-\mathbf{D}$ is closed in the topology of convergence {in measure}. 
Assume that $f_n\stackrel{\mathbb{F}}
\rightarrow f$, and $f_n\in\mathbf{B}(x)-\mathbf{D}$, $\forall n\geq 1$. There exists a sequence  $\{\mathbf{b}_n=\mathbf{b}(\cdot, X^n)\}_{n\geq1}\subset\mathbf{B}$ such that $\mathbf{b}_n-f_n\in{\mathbf{L}^0_+(\mathbb{D}, \mathcal{B}(\mathbb{D}), \mathbb{F})}$, $\forall n\geq 1$. Because $\{X^n_T\}_{n\geq1}\subset\mathcal{C}(x)$ defined in \cite{25} (see Section 3) and $\mathcal{C}(x)$ is closed in the topology of convergence almost surely under $\hat{Q}$, there exist $Z^n\in\text{conv}(X^n,X^{n+1},\dots)$, $\forall n\geq 1$, $Z_T\in\mathcal{C}(x)$ and $X\in\mathfrak{X}(x)$ such that $Z^n_T\rightarrow Z_T$, a.s. $\hat{Q}$ and $Z_T\leq X_T$, a.s. $\hat{Q}$. As such, by the concavity of $U$, we have
\begin{equation}\nonumber
\begin{split}
\mathbf{b}(\mu,X)&\geq\mathbf{E}^{Q^\mu}\left[U(Z_T)\right]\\
&\geq\liminf_{n\rightarrow\infty}\mathbf{E}^{Q^\mu}\left[U(Z^n_T)\right]\\
&\geq\liminf_{n\rightarrow\infty}\mathbf{E}^{Q^\mu}\left[U(X^n_T)\right]\\
&=\liminf_{n\rightarrow\infty}\mathbf{b}(\mu,X^n)\\
&\geq\liminf_{n\rightarrow\infty}f_n(\mu) \\
&=f(\mu) .
\end{split}
\end{equation}
Noting that $\mathbf{b}(\cdot,X)\in\mathbf{H}(x)\subset\mathbf{B}(x)-\mathbf{D}$, we have $f\in\mathbf{B}(x)-\mathbf{D}$, i.e., $\mathbf{B}(x)-\mathbf{D}$ is closed in the topology of convergence {in measure}.
}}

For any ${{\mathbf{b}(\cdot,\hat{X})}}\in \mathbf{B}(x)$, define $\mathbf{M}\triangleq{{\left\{\mathbf{b}(\cdot,\hat{X})\right\}}}+{\mathbf{L}^0_+(\mathbb{D}, \mathcal{B}(\mathbb{D}), \mathbb{F})}$. Then $\mathbf{M}$ is also a convex {{closed}} set and
\[\mathbf{M}\cap (\mathbf{B}(x)-\mathbf{D})={{\left\{\mathbf{b}(\cdot,\hat{X})\right\}}}.\]
 Based on convex set separation theorem, there exists a $\lambda\in\mathbf{L}^0_+(\mathbb{D}, \mathcal{B}(\mathbb{D}), \mathbb{F})$ with $ <\lambda,\mathbb{I}_{\mathbb{D}}>=1$ such that
\begin{equation}\nonumber
\begin{cases}
<{{\mathbf{b}(\cdot,\tilde{X})}}-{{\mathbf{b}(\cdot,\hat{X})}},\lambda(\cdot)>\leq 0,\\
<\lambda(\cdot),\mathbf{d}(\cdot)>\geq 0,
\end{cases}
\end{equation}
where ${{\mathbf{b}(\cdot,\tilde{X})}}\in \mathbf{B}(x)$ and $\mathbf{d}\in {\mathbf{L}^0_+(\mathbb{D}, \mathcal{B}(\mathbb{D}), \mathbb{F})}$ are arbitrary.  The first inequality implies
\[{{\hat{X}}} = \arg\max\limits_{{{X\in\mathfrak{X}(x)}}}\left\{ <{{\mathbf{b}(\cdot,X)}},\lambda(\cdot)>\right\},\]
and the second one implies $\lambda\in {\mathbf{L}^0_+(\mathbb{D}, \mathcal{B}(\mathbb{D}), \mathbb{F})}$.
\end{proof}

\vskip 5pt

\section{ {\bf{Proof of Proposition \ref{Prop1}.}} }\label{proof of Prop1}
\begin{proof}[{\bf{Proof of Proposition \ref{Prop1}}}]
{{For any $x>0$, we first prove that $\mathbf{b}(\cdot,xS^0)\in\mathbf{B}(x)$, 
i.e., $\forall X\in \mathfrak{X}(x)$, $X\neq xS^0$,}}
$\exists Q^{\mu_1} \in \mathcal{P}$ such that
\[{{\mathbf{b}(\mu_1,X)}}< U(xe^{rT}).\]
As $U(x)$ is strictly concave, for any $X\in\mathfrak{X}(x)$, we know
\begin{equation*}
{{\mathbf{b}(\mu,X)=}}\mathbf{E}^{Q^\mu}[U({{X_T}})]\leq U\left(\mathbf{E}^{Q^\mu}[{{X_T}}]\right),~\forall \mu\in\mathbb{D}.
\end{equation*}
Besides, ${{X_T}}$ satisfies the budget constraint $\mathbf{E}^{Q^*}[{{X_T}}]\leq xe^{rT}$. 
As such, if {{$Q^* \in \mathcal{P}$, let $Q^{\mu_1}=Q^*$}}, then
\begin{equation*}
\mathbf{E}^{Q^*}[U({{X_T}})]\leq U\left(\mathbf{E}^{Q^*}[{{X_T}}]\right)\leq U(xe^{rT}),
\end{equation*}
and the equality holds if and only if ${{X_T}}=\mathbf{E}^{Q}[{{X_T}}]=xe^{rT},~a.s.$, {{which means $X=xS^0$.}} 

Next, we prove that $h(x)=U(xe^{rT})$. Because $U(xe^{rT}) \in \mathbf{B}(x)$, there exists a $\lambda_0 \in \mathbf{S}$ such that
\begin{equation*}
<U(xe^{rT}),\lambda_0>~=~h(x)\rho(\lambda_0).
\end{equation*}
As $U(xe^{rT})$ is independent of $\mu$, we have
\[U(xe^{rT})<\mathbb{I}_{\mathbb{D}},\lambda_0>=h(x)\rho(\lambda_0).\]
Thus, $h(x)=U(xe^{rT})$.
\end{proof}

\vskip 5pt

{{

\section{\bf{ Proof of Proposition \ref{set}. } }\label{proof of set}
\begin{proof}[{ \bf{ Proof of Proposition \ref{set} } }]
{\bf\textit{  If side.}}
Following Definition \ref{DD2}, we have shown the fact that $\int_{\mathbb{D}}fg\mathrm{d}\mathbb{F}\leq h(x)y$ for any $f\in\mathbf{B}(x)-\mathbf{D}$, $g\in\mathbf{G}(y)$,  and the {\bf{if side}} naturally holds.

{\bf\textit{ Only if side.}}
{Let $f{\in\mathbf{L}^0_+(\mathbb{D},\mathcal{B}(\mathbb{D}),\mathbb{F})}$ and satisfy}
\begin{equation}\nonumber
\int_{\mathbb{D}}fg\mathrm{d}\mathbb{F}\leq h(x)y,\forall g\in\mathbf{G}(y).
\end{equation}
If $f\notin\mathbf{B}(x)-\mathbf{D}$, we have $\left(\mathbf{B}(x)-\mathbf{D}\right)\cap\left(\{f\}+\mathbf{L}^0_+(\mathbb{D}, \mathcal{B}(\mathbb{D}), \mathbb{F})\right)=\emptyset$. As $\mathbf{B}(x)-\mathbf{D}$ and $\{f\}+\mathbf{L}^0_+(\mathbb{D}, \mathcal{B}(\mathbb{D}), \mathbb{F})$ are both closed and convex, based on the convex set separation theorem, there exists a $\lambda\in\mathbf{S}$ such that, for any $f'\in\mathbf{B}(x)-\mathbf{D}$, $<f'-f,\lambda><0$. Noting that $\mathbf{b}(\cdot,\hat{X}^\lambda)\in \mathbf{B}(x)-\mathbf{D}$ and $\frac{y\lambda}{\rho(\lambda)}\in\mathbf{G}(y)$, we have
\begin{equation}\nonumber
\begin{split}
\int_{\mathbb{D}}f\frac{y\lambda}{\rho(\lambda)}\mathrm{d}\mathbb{F}&=<f,\frac{y\lambda}{\rho(\lambda)}>\\
&><\mathbf{b}(\cdot,\hat{X}^\lambda),\frac{y\lambda}{\rho(\lambda)}>\\
&=\frac{y}{\rho(\lambda)}h(x)\rho(\lambda)\\
&=h(x)y,
\end{split}
\end{equation}
which contradicts to the condition that $\int_{\mathbb{D}}fg\mathrm{d}\mathbb{F}\leq h(x)y,\forall g\in\mathbf{G}(y)$. As such, $f\in\mathbf{B}(x)-\mathbf{D}$.

{Let $g{\in\mathbf{L}^0_+(\mathbb{D},\mathcal{B}(\mathbb{D}),\mathbb{F})}$ and satisfy}
\begin{equation}\nonumber
\int_{\mathbb{D}}fg\mathrm{d}\mathbb{F}\leq h(x)y,\forall f\in\mathbf{B}(x)-\mathbf{D}.
\end{equation}
{Define} $g^*=\frac{g}{<g,\mathbb{I}_{\mathbb{D}}>}$. We have $\mathbf{b}(\cdot,\hat{X}^{g^*})\in\mathbf{B}(x)-\mathbf{D}$, then
\begin{equation}\nonumber
\begin{split}
\rho(g^*)&=\frac{h(x)\rho(g^*)}{h(x)}\\
&=\frac{1}{h(x)}\int_{\mathbb{D}}\mathbf{b}(\cdot,\hat{X}^{g^*})g^*\mathrm{d}\mathbb{F}\\
&=\frac{1}{h(x)<g,\mathbb{I}_{\mathbb{D}}>}\int_{\mathbb{D}}\mathbf{b}(\cdot,\hat{X}^{g^*})g\mathrm{d}\mathbb{F}\\
&\leq \frac{1}{h(x)<g,\mathbb{I}_{\mathbb{D}}>}h(x)y\\
&=\frac{y}{<g,\mathbb{I}_{\mathbb{D}}>}.
\end{split}
\end{equation}
As such, $g=<g,\mathbb{I}_{\mathbb{D}}>g^*\leq\frac{yg^*}{\rho(g^*)}$ and $<g^*,\mathbb{I}_\mathbb{D}>=1$, which means $g\in\mathbf{G}(y)$.
\end{proof}


\section{\bf{ Proof of Theorem \ref{duality}}.}\label{proof of duality}
Observing that the results in Theorem \ref{duality} are similar to Theorem 3.1 in \cite{25} with the only difference that $x$ is replaced by $h(x)$, we break the proof into several lemmas as well. First, we prove that $\mathbf{G}(y)$ is convex and closed in the topology of convergence almost surely.
\begin{lemma}\label{l1}
For any $y>0$, $\mathbf{G}(y)$ is convex and closed in the topology of convergence in measure.
\end{lemma}
\begin{proof}
{\bf {Convexity.}} For $g_1$, $g_2\in\mathbf{G}(y)$, there exist $\lambda_1$, $\lambda_2\in\mathbf{S}$ such that $g_1\leq \frac{y\lambda_1}{\rho(\lambda_1)}$ and $g_2\leq \frac{y\lambda_2}{\rho(\lambda_2)}$. Recall that $\rho$ is convex, and for any $\alpha\in(0,1)$, we have
\begin{equation}\nonumber
\begin{split}
\rho\left(\alpha\frac{\lambda_1}{\rho(\lambda_1)}+(1-\alpha)\frac{\lambda_2}{\rho(\lambda_2)}\right)&\leq\alpha\rho\left(\frac{\lambda_1}{\rho(\lambda_1)}\right)+(1-\alpha)\rho\left(\frac{\lambda_2}{\rho(\lambda_2)}\right)\\
&=\alpha\frac{\rho(\lambda_1)}{\rho(\lambda_1)}+(1-\alpha)\frac{\rho(\lambda_2)}{\rho(\lambda_2)}\\
&=1.
\end{split}
\end{equation}
As such, denote $\lambda^*=\alpha\frac{\lambda_1}{\rho(\lambda_1)}+(1-\alpha)\frac{\lambda_2}{\rho(\lambda_2)}$, and we have $\rho(\lambda^*)\leq 1$ and
\begin{equation}\nonumber
\begin{split}
\alpha g_1+(1-\alpha)g_2&\leq\alpha\frac{y\lambda_1}{\rho(\lambda_1)}+(1-\alpha)\frac{y\lambda_2}{\rho(\lambda_2)}\\
&=y\lambda^*\\
&\leq \frac{y\lambda^*}{\rho(\lambda^*)},
\end{split}
\end{equation}
which means that $\alpha g_1+(1-\alpha)g_2\in\mathbf{G}(y)$, i.e., $\mathbf{G}(y)$ is convex.


{\bf {Closed.}} Based on the bipolar relation in Theorem \ref{bipolarThm}, we have
\begin{equation}\nonumber
\begin{split}
&\mathbf{G}^0=\frac{1}{U(e^{rT})}\mathbf{B}(1)-\mathbf{D},\\
&\left(\frac{1}{U(e^{rT})}\mathbf{B}(1)-\mathbf{D}\right)^0=\mathbf{G}.
\end{split}
\end{equation}
As such, $\mathbf{G}^{0}=\mathbf{G}$, which means that $\mathbf{G}$ is closed in the topology of convergence {in measure}. For any $y>0$, $\mathbf{G}(y)=y\mathbf{G}$, which is also closed in the topology of convergence {in measure}.

\end{proof}

Denote by $\psi^+$ and $\psi^-$ the positive and negative parts of the function $\psi$, respectively.
\begin{lemma}\label{l2}
Under the assumptions of Theorem \ref{duality}, for any $y>0$, the family $\{\psi^-(g)):g\in\mathbf{G}(y)\}$ is uniformly integrable, and if $\{g_n\}_{n\geq1}$ is a sequence in $\mathbf{G}(y)$ which converges almost surely to a random variable $g$, then $g\in\mathbf{G}(y)$ and
\begin{equation}\label{uniform}
\liminf\limits_{n\rightarrow\infty}\int_{\mathbb{D}}\psi(g_n)\mathrm{d}\mathbb{F}\geq\int_{\mathbb{D}}\psi(g)\mathrm{d}\mathbb{F}.
\end{equation}
\end{lemma}

\begin{proof}
Assume that $\psi(\infty)< 0$ (otherwise there is nothing to prove). Let $L:(-\psi(0),-\psi(\infty))\rightarrow (0,\infty)$ denote the inverse of $-\psi$. The function $L$ is strictly increasing, and for any $g\in\mathbf{G}(y)$,
\begin{equation}\nonumber
\begin{split}
\int_{\mathbb{D}}L\left(\psi^-(g)\right)\mathrm{d}\mathbb{F}&\leq\int_{\mathbb{D}}L\left(-\psi(g)\right)\mathrm{d}\mathbb{F}+L(0)\\
&=\int_{\mathbb{D}}g\mathrm{d}\mathbb{F}+L(0)\\
&\leq y+L(0),
\end{split}
\end{equation}
and by the Inada condition of $\phi$ and the l'Hospital rule
\begin{equation}\nonumber
\lim\limits_{x\rightarrow-\psi(\infty)}\frac{L(x)}{x}=\lim\limits_{y\rightarrow\infty}\frac{y}{-\psi(y)}=\lim\limits_{y\rightarrow\infty}\frac{1}{\left(\phi'\right)^{-1}(y)}=\infty.
\end{equation}
As $\mathbf{G}(y)$ is bounded in $\mathbf{L}^1(\mathbb{D},\mathcal{B}(\mathbb{D}),\mathbb{F})$, by applying the de la Vall\'{e}e–Poussin theorem, we have the uniform integrability of $\{\psi^-(g)):g\in\mathbf{G}(y)\}$.

Let $\{g_n\}_{n\geq 1}$ be a sequence in $\mathbf{G}(y)$ which converges almost surely to a variable $g$. It follows from the uniform integrability of the sequence $\{\psi^-(g_n))\}_{n\geq 1}$ that
\[\lim\limits_{n\rightarrow\infty}\int_{\mathbb{D}}L\left(\psi^-(g_n)\right)\mathrm{d}\mathbb{F}=\int_{\mathbb{D}}\psi^-(g)\mathrm{d}\mathbb{F}\]
and from Fatou's lemma that
\[\liminf\limits_{n\rightarrow\infty}\int_{\mathbb{D}}L\left(\psi^+(g_n)\right)\mathrm{d}\mathbb{F}\geq\int_{\mathbb{D}}\psi^+(g)\mathrm{d}\mathbb{F},\]
which implies Eq.~(\ref{uniform}). As $\mathbf{G}(y)$ is closed under convergence in probability, we know $g\in\mathbf{G}(y)$.
\end{proof}

We are now able to prove assertion (ii) of Theorem \ref{duality}.
\begin{lemma}\label{l3}
In addition to the assumptions of Theorem \ref{duality}, assume $v(y)<\infty$. Then the optimal solution $\hat{g}(y)$ to Problem~(\ref{D2}) exists and is unique. As a consequence, $v(y)$ is strictly convex on $\{v<\infty\}$.
\end{lemma}

\begin{proof}
Let $\{g_n\}_{n\geq1}$ be a sequence in $\mathbf{G}(y)$ such that
\[\lim\limits_{n\rightarrow\infty}\int_{\mathbb{D}}\psi(g_n)\mathrm{d}\mathbb{F}=v(y).\]
There exists a sequence $h_n\in~\text{conv}(g_n,g_{n+1},\dots)$, $n\geq1$, and a variable $\hat{h}$ such that $h_n\rightarrow\hat{h}$ a.s. $\mathbb{F}$. From the convexity of the function,
$\psi$ we deduce that
\[\int_{\mathbb{D}}\psi(h_n)\mathrm{d}\mathbb{F}\leq\sup\limits_{m\geq n}\int_{\mathbb{D}}\psi(g_m)\mathrm{d}\mathbb{F},\]
then
\[\lim\limits_{n\rightarrow\infty}\int_{\mathbb{D}}\psi(h_n)\mathrm{d}\mathbb{F}=v(y).\]
As such
\[\int_{\mathbb{D}}\psi(\hat{h})\mathrm{d}\mathbb{F}\leq\lim\limits_{n\rightarrow\infty}\int_{\mathbb{D}}\psi(h_n)\mathrm{d}\mathbb{F}=v(y),\]
and $\hat{h}\in\mathbf{G}(y)$, which implies that $\hat{h}$ is a solution to Problem~(\ref{D2}). The uniqueness of the optimal solution follows from the strict convexity of the function $\psi$.
Using again the strict convexity of $\psi$, for $y_1<y_2$ with $v(y_1)<\infty$, note that $(\hat{h}(y_1)+\hat{h}(y_2))/2\in\mathbf{G}((y_1+y_2)/2)$, and we have
\[v\left(\frac{y_1+y_2}{2}\right)\leq\int_{\mathbb{D}}\psi\left(\frac{\hat{h}(y_1)+\hat{h}(y_2)}{2}\right)\mathrm{d}\mathbb{F}<\frac{v(y_1)+v(y_2)}{2}.\]
\end{proof}

We now turn to the proof of assertion (i) of Theorem \ref{duality}. Since the value function $u$  clearly is concave and $u(x_0)< \infty$, for some $x_0 > 0$, we have $u(x) < \infty$, for all $x > 0$.
\begin{lemma}\label{l4}
Under the assumptions of Theorem \ref{duality}, we have the fact that the value functions $u$ and $v$ are distorted conjugate, i.e., for any $x>0$, $y>0$,
\begin{equation}\label{dual22}
\begin{cases}
u(x)=\mathop {\inf}\limits_{y>0}[v(y)+h(x)y],\\
v(y)=\mathop {\sup}\limits_{x>0}[u(x)-h(x)y].
\end{cases}
\end{equation}

\end{lemma}
}}

\begin{proof}
Because $h(x)$ is invertible, the two equations in Eq.~(\ref{dual22}) are equivalent, and we only prove the latter one.

For $n>0$, we define $\mathscr{B}_n$ as follows:
\begin{equation}\nonumber
\mathscr{B}_n=\{f:0\leq f\leq n\}.
\end{equation}
The sets $\mathscr{B}_n$, $n>0$, are $\sigma(L^\infty,L^1)$-compact. 
Noting that $\mathbf{G}(y)$ is a closed convex subset of $\mathbf{L}^1(\mathbb{D},\mathcal{B}(\mathbb{D}),\mathbb{F})$. 
Then, using mini-max theorem (see  Theorem 45.8 in \cite{maxmin}), we have, for fixed $n$,
\begin{equation}\nonumber
\sup_{f \in \mathscr{B}_n}\inf_{g\in \mathbf{G}(y)}\left\{\int_{\mathbb{D}}\left[\phi(f)-fg\right]\mathrm{d}\mathbb{F}\right\}
=\inf_{g\in \mathbf{G}(y)}\sup_{f\in\mathscr{B}_n}\left\{\int_{\mathbb{D}}
\left[\phi(f)-fg\right]\mathrm{d}\mathbb{F}\right\}.
\end{equation}
Noting
\begin{equation}\nonumber
\frac{n\mathbb{I}_{\mathbb{D}}}{\rho(\mathbb{I}_{\mathbb{D}})}\in\mathbf{G}(n),
\end{equation}
 we know
\begin{equation}\nonumber
\left\{
	\begin{split}
		&\mathscr{B}_n\subset\rho(\mathbb{I}_{\mathbb{D}})\mathbf{G}(n)=\mathbf{G}(n\rho(\mathbb{I}_{\mathbb{D}})),\\
		&\mathscr{B}_n\subset\mathbf{B}\left(h^{-1}\left(\frac{n}{\rho(\mathbb{I}_{\mathbb{D}})}\right)\right)-\mathbf{D}.
	\end{split}
\right.
\end{equation}

Based on Proposition \ref{set},
\begin{equation}\nonumber
\begin{split}
&\lim_{n\rightarrow\infty}\sup_{f \in \mathscr{B}_n}\inf_{g\in \mathbf{G}(y)}\left\{\int_{\mathbb{D}}\left[\phi(f)-fg\right]\mathrm{d}\mathbb{F}\right\}\\
&=\sup_{x>0}\sup_{f\in \mathbf{B}(x)-\mathbf{D}}\left\{\int_{\mathbb{D}}\left[\phi(f)-h(x)y\right]\mathrm{d}\mathbb{F}\right\}\\
&=\sup_{x>0}[u(x)-h(x)y].
\end{split}
\end{equation}
Besides,
\begin{equation}\nonumber
\inf_{g\in \mathbf{G}(y)}\sup_{f\in\mathscr{B}_n}\left\{\int_{\mathbb{D}}
\left[\phi(f)-fg\right]\mathrm{d}\mathbb{F}\right\}=\inf_{g\in \mathbf{G}(y)}\left\{\int_{\mathbb{D}}\psi_n(g)\mathrm{d}\mathbb{F}\right\}\triangleq v_n(y),
\end{equation}
where
\[\psi_n(y)=\sup_{0<x\leq n}[\phi(x)-xy].\]
Consequently, it is sufficient to show
\begin{equation}\label{F5}
\lim_{n\rightarrow\infty}v_n(y)=\lim_{n\rightarrow\infty}\sup_{g\in \mathbf{G}(y)}\left\{\int_{\mathbb{D}}\psi_n(g)\mathrm{d}\mathbb{F}\right\}=v(y),~~y>0.
\end{equation}
Evidently, $\{v_n\}_{n\geq 1}$ is a increasing sequence, and $v_n\leq v$ for $n\geq 1$. Let $\{g_n\}_{n\geq 1}\subset\mathbf{B}(x)-\mathbf{D}$ such that
\begin{equation}\nonumber
\lim_{n\rightarrow\infty}\int_{\mathbb{D}}\psi_n(g_n)\mathrm{d}\mathbb{F}=\lim_{n\rightarrow\infty}v_n(y).
\end{equation}

As $\{g_n\}_{n\geq1}$ is a sequence of non-negative random variables, there always exists a sequence $h_n\in conv(g_n,g_{n+1},\dots)$ such that  $h_n\mathop{\longrightarrow} \limits^{a.s.} h$. As $\mathbf{G}(y)$ is closed, we have $h\in \mathbf{G}(y)$. For $y\geq \left(\phi'\right)^{-1}(1)\geq \left(\phi'\right)^{-1}(n)$, $\psi_n(y)=\psi(y)$. Based on Lemma \ref{l2}, the sequence $\{\psi_n^-(h_n)\}_{n\geq 1}$ is uniformly integrable. The convexity of $\psi_n$ and Fatou's lemma imply
\begin{equation}\nonumber
\lim_{n\rightarrow\infty}\int_{\mathbb{D}}\psi_n(g_n)\mathrm{d}\mathbb{F}\geq\liminf_{n\rightarrow\infty}\int_{\mathbb{D}}\psi_n(h_n)\mathrm{d}\mathbb{F}\geq\int_{\mathbb{D}}\psi(h)\mathrm{d}\mathbb{F}\geq v(y).
\end{equation}
Thus  Eq.~(\ref{F5}) follows.
\end{proof}

{{
\begin{lemma}\label{l5}
Under the assumptions of Theorem \ref{duality}, we have
\begin{equation}\label{i1}
\mathop {\lim}\limits_{x\rightarrow 0}u'(x)=~\infty,~~\mathop {\lim}\limits_{y\rightarrow \infty}v'(y)=~0.
\end{equation}
\end{lemma}

\begin{proof}
By the duality relation (\ref{dual22}), the derivatives $u'$ and $v'$ of the value functions $u$ and $v$ satisfy
\begin{equation}\nonumber
\begin{split}
&-v'(y)=\inf\left\{h(x):\frac{u'(x)}{h'(x)}\leq y\right\},~y>0,\\
&\frac{u'(x)}{h'(x)}=\inf\left\{y:-v'(y)\leq h(x)\right\},~x>0.
\end{split}
\end{equation}
As $h(0)<\infty$ and $h'(0)=\infty$, it follows that the assertions of Eq.~(\ref{i1}) are equivalent. We shall prove the second one. The function $-v$ is concave and increasing. Hence there is a finite positive limit
\[-v'(\infty)\triangleq\lim\limits_{y\rightarrow\infty}-v'(y).\]
Because the function $-\psi$ is increasing and $-\psi'(y)=\left(\phi'\right)^{-1}(y)$ tends to $0$ as $y$ tends to $\infty$, for any $\epsilon > 0$ there exists a number $C(\epsilon)$ such that
\[-\psi(y)\leq C(\epsilon)+\epsilon y~\text{     }~\forall y>0.\]
As $\mathbf{G}(y)$ is bounded in $\mathbf{L}^1(\mathbb{D},\mathcal{B}(\mathbb{D}),\mathbb{F})$, we have
\begin{equation}\nonumber
\begin{split}
0&\leq-v'(\infty)=\lim\limits_{y\rightarrow\infty}\frac{-v(y)}{y}=\lim\limits_{y\rightarrow\infty}\sup\limits_{g\in\mathbf{G}(y)}\int_{\mathbb{D}}\frac{-\psi(g)}{y}\mathrm{d}\mathbb{F}\\
&\leq\lim\limits_{y\rightarrow\infty}\sup\limits_{g\in\mathbf{G}(y)}\int_{\mathbb{D}}\left[\frac{C(\epsilon)+\epsilon g}{y}\right]\mathrm{d}\mathbb{F}
\leq\lim\limits_{y\rightarrow\infty}\sup\limits_{g\in\mathbf{G}(y)}\int_{\mathbb{D}}\left[\frac{C(\epsilon)}{y}+\epsilon\right]\mathrm{d}\mathbb{F}=\epsilon.
\end{split}
\end{equation}
Consequently, $-v'(\infty)=0$.
\end{proof}

In the setting of Theorem \ref{duality}, we show the following
results for later use.

\begin{lemma}\label{l6}
Under the assumptions of Theorem \ref{duality}, let $\{y_n\}_{n\geq 1}$ be a sequence of positive numbers which converges to a number $y>0$ and satisfies  $v(y_n)<\infty$ and $v(y)<\infty$. Then $\hat{g}(y_n)$ converges to $\hat{g}(y)$ in probability and $\psi\left(\hat{g}(y_n)\right)$ converges to $\psi\left(\hat{g}(y)\right)$ in $\mathbf{L}^1(\mathbb{D},\mathcal{B}(\mathbb{D}),\mathbb{F})$. 
\end{lemma}
\begin{proof}
If $\hat{g}(y_n)$ does not converge to $\hat{g}(y)$ in probability, then there exists
$\epsilon > 0$ such that
\[\limsup\limits_{n\rightarrow\infty}\mathbb{F}(|\hat{g}(y_n)-\hat{g}(y)|>\epsilon)>\epsilon.\]
Moreover, as $U(e^{rT})\in\mathbf{B}(1)-\mathbf{D}$ and the bipolar relation (\ref{bipolar}), we know that $\int_{\mathbb{D}}\hat{g_n}\mathrm{d}\mathbb{F}\leq1$ and $\int_{\mathbb{D}}\hat{g}\mathrm{d}\mathbb{F}\leq1$, then by possibly passing to a smaller $\epsilon > 0$, we  assume that
\begin{equation}\label{epsilon}
\limsup\limits_{n\rightarrow\infty}\mathbb{F}\left(|\hat{g}(y_n)+\hat{g}(y)|<1/\epsilon;|\hat{g}(y_n)-\hat{g}(y)|>\epsilon\right)>\epsilon.
\end{equation}
Define
\[g_n=\frac{1}{2}(\hat{g}(y_n)+\hat{g}(y)),~n\geq 1.\]
By the convexity of $\psi$, we have
\[\psi(h_n)\leq\frac{1}{2}\left[\psi(\hat{g}(y_n))+\psi(\hat{g}(y))\right],\]
and by Eq.~(\ref{epsilon}), we deduce the existence of $\delta>0$ such that
\[\limsup\limits_{n\rightarrow\infty}\mathbb{F}\left\{\psi(h_n)\leq\frac{1}{2}\left[\psi(\hat{g}(y_n))+\psi(\hat{g}(y))\right]-\delta\right\}>\delta.\]
As such
\begin{equation}\nonumber
\begin{split}
\int_{\mathbb{D}}\psi(g_n)\mathrm{d}\mathbb{F}&\leq\frac{1}{2}\left[\int_{\mathbb{D}}\psi(\hat{g}(y_n))\mathrm{d}\mathbb{F}+\int_{\mathbb{D}}\psi(\hat{g})(y))\mathrm{d}\mathbb{F}\right]-\delta^2\\
&=\frac{1}{2}\left(v(y_n)+v(y)\right)-\delta^2.
\end{split}
\end{equation}
The function $v$ is convex and therefore continuous on the set $\{v<\infty\}$. It follows that
\[\limsup\limits_{n\rightarrow\infty}\int_{\mathbb{D}}\psi(g_n)\mathrm{d}\mathbb{F}\leq v(y)-\delta^2.\]
Then there exists a sequence $h_n\in\text{conv}(g_n,g_{n+1},\dots)$, $n\geq1$, which converges almost surely to a variable $h$. It follows from Lemma \ref{l2} and the convexity of $\psi$ that $h\in\mathbf{G}(y)$ and
\begin{equation}\nonumber
\begin{split}
\int_{\mathbb{D}}\psi(h)\mathrm{d}\mathbb{F}&=\int_{\mathbb{D}}\liminf\limits_{n\rightarrow\infty}\psi(h_n)\mathrm{d}\mathbb{F}\leq\liminf\limits_{n\rightarrow\infty}\int_{\mathbb{D}}\psi(h_n)\mathrm{d}\mathbb{F}\\
&{\leq\liminf\limits_{n\rightarrow\infty}\int_{\mathbb{D}}\psi(g_n)\mathrm{d}\mathbb{F}}\leq v(y)-\delta^2,
\end{split}
\end{equation}
which contradicts with the definition of $v(y)$. As such, $\hat{g}(y_n)$ converges to $\hat{g}(y)$ in probability. 
Lemma \ref{l2} shows that $\left\{\psi^-\left(\hat{g}(y_n)\right)\right\}_{n\geq1}$ is uniformly integrable. Consequently, $\psi(\hat{g}(y_n))$ converges to $\psi(\hat{g}(y))$ in $\mathbf{L}^1(\mathbb{D},\mathcal{B}(\mathbb{D}),\mathbb{F})$ if
\[\lim\limits_{n\rightarrow\infty}\int_{\mathbb{D}}\psi(\hat{g}(y_n))\mathrm{d}\mathbb{F}\leq \psi(\hat{g}(y)),\]
which in turn follows from the continuity of the value function $v$ on the set $\{v<\infty\}$.
\end{proof}

Based on Lemma \ref{l4}, the continuous differentiability and the concavity of $u$ follow from the strict convexity of $v$ on $\{v<\infty\}$ by general duality results, which prove the rest of Theorem \ref{duality}.

}}

{{

\section{ { \bf{Proof of Theorem \ref{solution}.} } }\label{proof of solution}
Again, observing that the results in Theorem \ref{duality} are similar to Theorem 3.1 in \cite{25}, the proof of Theorem \ref{solution} will be broken into several steps. Without loss of generality, we assume that $\phi(\infty)=\psi(0)>0$.
\begin{lemma}\label{l7}
Under the assumptions of Theorem \ref{solution}, let $\{y_n\}_{n\geq 1}$ be sequence of positive numbers {converging} to a number $y>0$. Then $\hat{g}(y_n)\psi'(\hat{g}(y_n))$ converges to $\hat{g}(y)\psi'(\hat{g}(y))$ in $\mathbf{L}^1(\mathbb{D},\mathcal{B}(\mathbb{D}),\mathbb{F})$.
\end{lemma}
\begin{proof}

Using Lemma \ref{l6}, we know that $\hat{g}(y_n)$ converges to $\hat{g}(y)$ in probability. As such, based on  the continuity of $\psi'$, we know $\hat{g}(y_n)\psi'(\hat{g}(y_n))$ converges to $\hat{g}(y)\psi'(\hat{g}(y))$ in probability.

In order to obtain the conclusion we have to show the uniform integrability of the sequence $\hat{g}(y_n)\psi'(\hat{g}(y_n))$. At this point, we use the hypothesis that the asymptotic elasticity of $\phi$ is less then one, which implies the existence of $y_0 > 0$ and a constant $C < \infty$ such that
\[-\psi'(y)<C\frac{\psi(y)}{y},~\text{for}~0<y<y_0.\]

Hence the sequence of random variables $\{|\hat{g}(y_n)\psi'(\hat{g}(y_n))\mathbb{I}_{\{\hat{g}(y_n)<y_0\}}|\}_{n\geq1}$ is dominated by the sequence $\{C|\psi(\hat{g}(y_n))|\mathbb{I}_{\{\hat{g}(y_n)<y_0\}}\}_{n\geq1}$ which is uniformly integrable by Lemma \ref{l6}.

For the remaining part $\{\hat{g}(y_n)\psi'(\hat{g}(y_n))\mathbb{I}_{\{\hat{g}(y_n)\geq y_0\}}\}_{n\geq1}$, the uniform integrability follows as in the proof of Lemma \ref{l2} from the fact that $\{\hat{g}(y_n)\}_{n\geq1}$ is bounded in $\mathbf{L}^1(\mathbb{D},\mathcal{B}(\mathbb{D}),\mathbb{F})$ and $\lim\limits_{y\rightarrow\infty}\psi'(y)=0$.

\end{proof}

\begin{remark}\label{mu}
For later use, we remark that, given the setting of Lemma \ref{l7} and in addition a sequence $\{\mu_n\}_{n\geq1}$ of real numbers tending to $1$, we still conclude that $\psi'\left(\mu_n\hat{g}(y_n)\right)\hat{g}(y_n)$ tends to $\psi'\left(\hat{g}(y)\right)\hat{g}(y)$ in $\mathbf{L}^1(\mathbb{D},\mathcal{B}(\mathbb{D}),\mathbb{F})$.
Indeed, it suffices to note that from $AE(\phi)<1$, for fixed $0<\mu<1$, we can find a constant $\tilde{C}<\infty$ and $y_0 > 0$ such that
\[-\psi'(\mu y)<\tilde{C}\frac{V(y)}{y},~\text{for}~0<y<y_0.\]
\end{remark}

\begin{lemma}\label{l8}
Under the assumptions of Theorem \ref{solution}, the value function $v$ is finite and continuously differentiable on $(0,\infty)$, the derivative $v'$ is strictly increasing and satisfies
\begin{equation}\nonumber
-yv'(y)=\int_{\mathbb{D}}\hat{g}(y)\left(\phi'\right)^{-1}(\hat{g}(y))\mathrm{d}\mathbb{F}.
\end{equation}
\end{lemma}

\begin{proof}
Observe that, provided the limit exists,
\begin{equation}\label{derivative}
-yv'(y)=\lim\limits_{\alpha\searrow 1}\frac{v(y)-v(\alpha y)}{\alpha-1}.
\end{equation}
By Lemma \ref{l7} and convexity of $\psi$, we have
\begin{equation*}
\begin{split}
\limsup_{\alpha\searrow 1}\frac{v(y)-v(\alpha y)}{\alpha-1}&\leq\limsup_{\alpha\searrow 1}\frac{1}{\alpha-1}\int_{\mathbb{D}}\left[\psi\left(\frac{1}{\alpha}\hat{g}(\alpha y)\right)-\psi\left(\hat{g}(\alpha y)\right)\right]\mathrm{d}F(\mu)\\
&\leq\limsup_{\alpha\searrow 1}\frac{1}{\alpha-1}\int_{\mathbb{D}}\left[\left(\frac{1}{\alpha}-1\right)\hat{g}(\alpha y)\psi'\left(\frac{1}{\alpha}\hat{g}(\alpha y)\right)\right]\mathrm{d}F(\mu)\\
&=\int_{\mathbb{D}}\hat{g}(y)\left(\phi'\right)^{-1}(\hat{g}(y))\mathrm{d}\mathbb{F},
\end{split}
\end{equation*}
where in the last line we have used Remark \ref{mu}.

By the monotone convergence theorem, we have
\begin{equation*}
\begin{split}
\liminf_{\alpha\searrow 1}\frac{v(y)-v(\alpha y)}{\alpha-1}&\geq\liminf_{\alpha\searrow 1}\frac{1}{\alpha-1}\int_{\mathbb{D}}\left[\psi\left(\hat{g}(y)\right)-\psi\left(\alpha \hat{g}(y)\right)\right]\mathrm{d}F(\mu)\\
&\geq\liminf_{\alpha\searrow 1}\frac{1}{\alpha-1}\int_{\mathbb{D}}\left[(1-\alpha)\hat{g}(y)\psi'\left(\hat{g}(\alpha y)\right)\right]\mathrm{d}F(\mu)\\
&=\int_{\mathbb{D}}\hat{g}(y)\left(\phi'\right)^{-1}(\hat{g}(y))\mathrm{d}\mathbb{F}.
\end{split}
\end{equation*}
This shows that Eq.~(\ref{derivative}) holds with $v'(y)$ replaced by the right derivative $v_r'(y)$. Using Lemma \ref{l7}, we obtain the continuity of the function
$y\rightarrow v_r'(y)$. The convexity of $v$ implies the continuous differentiable property of $v$ and we complete the proof.
\end{proof}

}}

\begin{lemma}\label{l9}
Under the assumptions of Theorem \ref{solution}, suppose that the numbers $x$ and $y$ are related by $h(x)=-v'(y)$. Then $\hat{f}(x)\triangleq \left(\phi'\right)^{-1}\left(\hat{g}(y)\right)$ is the unique solution to Problem~(\ref{p3}).
\end{lemma}

{{
\begin{proof}
We first show that $\hat{f}(x)\triangleq \left(\phi'\right)^{-1}\left(\hat{g}(y)\right)$ belongs to $\mathbf{B}(x)-\mathbf{D}$. According to the bipolar relation (\ref{bipolar}),  it is sufficient to show that, for any $g\in\mathbf{G}(y)$,
\begin{equation}\label{admissable}
\int_{\mathbb{D}}\left[g\left(\phi'\right)^{-1}(\hat{g})\right]\mathrm{d}\mathbb{F}\leq h(x)y=-yv'(y)=\int_{\mathbb{D}}\left[\hat{g}\left(\phi'\right)^{-1}(\hat{g})\right]\mathrm{d}\mathbb{F},
\end{equation}
where the last equality follows from Eq.~(\ref{derivative}).

Fix $g\in\mathbf{G}(y)$ and denote
\[g_\delta=(1-\delta)\hat{g}(y)+\delta g,~\delta\in(0,1).\]
From the inequality (observing that $\left(\phi'\right)^{-1}=-\psi'$)
\begin{equation}\nonumber
\begin{split}
0\leq\int_{\mathbb{D}}\psi(g_\delta)\mathrm{d}\mathbb{F}-\int_{\mathbb{D}}\psi(\hat{g})\mathrm{d}\mathbb{F}
=\int_{\mathbb{D}}\left[\int_{g_\delta}^{\hat{g}}-\psi'(z)\mathrm{d}z\right]\mathrm{d}\mathbb{F}
\leq\int_{\mathbb{D}}\left[-\psi'(g_\delta)(\hat{g}-g_\delta)\right]\mathrm{d}\mathbb{F},
\end{split}
\end{equation}
we deduce that
\begin{equation}\label{exy}
\int_{\mathbb{D}}\left[-\psi'\left((1-\delta)\hat{g}\right)\hat{g}\right]\mathrm{d}\mathbb{F}\geq\int_{\mathbb{D}}\left[-\psi'\left(g_\delta\right)g\right]\mathrm{d}\mathbb{F}.
\end{equation}
Remark \ref{mu} implies that for $\delta$ close to 0,
\[\int_{\mathbb{D}}\left[-\psi'\left((1-\delta)\hat{g}\right)\hat{g}\right]\mathrm{d}\mathbb{F}<\infty.\]
{By applying} the monotone convergence theorem and the Fatou lemma, respectively, to the left- and right-hand sides of (\ref{exy}), as $\delta\rightarrow0$, {we have} the desired inequality (\ref{admissable}). Hence, $\hat{f}(x)\in\mathbf{B}(x)-\mathbf{D}$.

For any $f\in\mathbf{B}(x)-\mathbf{D}$, we have
\begin{equation}\nonumber
\begin{split}
&\int_{\mathbb{D}}f\hat{g}(y)\mathrm{d}\mathbb{F}\leq h(x)y,\\
&\phi(f)\leq \psi(\hat{g}(y))+f\hat{g}(y).
\end{split}
\end{equation}
It follows that
\begin{equation}\nonumber
\begin{split}
\int_{\mathbb{D}}\phi(f)\mathrm{d}\mathbb{F}&\leq v(y)+h(x)y=\int_{\mathbb{D}}\left[\psi(\hat{g})-\hat{g}\psi'(\hat{g})\right]\mathrm{d}\mathbb{F},\\
&=\int_{\mathbb{D}}\phi\left(-\psi'(\hat{g})\right)\mathrm{d}\mathbb{F}=\int_{\mathbb{D}}\phi(\hat{f})\mathrm{d}\mathbb{F},
\end{split}
\end{equation}
which proves the optimality of $\hat{g}(x)$. The uniqueness of the optimal solution follows
from the strict concavity of the function $\phi$.
\end{proof}

\begin{proof}[{ \bf{Proof of Theorem \ref{solution}} }]
We have to check that the above lemmas imply all the assertions of Theorem \ref{solution}.
As regards the assertions
\[u'(\infty)=\lim_{x\rightarrow\infty}u'(x)=0~\text{and}~-v'(0)=\lim_{x\rightarrow 0}-v'(y)=\infty,\]
they are equivalent {as $-v'(y)$ is the inverse function of $\frac{u'}{h'}\!\circ\!h^{-1}(x)$} by Theorem \ref{bipolarThm} (i) and Lemma \ref{l8}. Hence it suffices to prove the first one. Similar with the proof of Lemma \ref{l5}, the function $u$ is concave and increasing. Hence there is a finite positive limit
\[u'(\infty)\triangleq\lim\limits_{x\rightarrow\infty}u'(x).\]
Because the function $\phi$ is increasing and $\phi'(x)$ tends to $0$ as $x$ tends to $\infty$, for any $\epsilon > 0$, there exists a number $C(\epsilon)$ such that
\[\phi(x)\leq C(\epsilon)+\epsilon x~\text{     }~\forall x>0.\]
As $\mathbf{B}(x)-\mathbf{D}$ is bounded in $\mathbf{L}^1(\mathbb{D},\mathcal{B}(\mathbb{D}),\mathbb{F})$, we have
\begin{equation}\nonumber
\begin{split}
0&\leq u'(\infty)=\lim\limits_{x\rightarrow\infty}\frac{u(x)}{x}=\lim\limits_{x\rightarrow\infty}\sup\limits_{f\in\mathbf{B}(x)-\mathbf{D}}\int_{\mathbb{D}}\frac{\phi(f)}{x}\mathrm{d}\mathbb{F}\\
&\leq\lim\limits_{x\rightarrow\infty}\sup\limits_{f\in\mathbf{B}(x)-\mathbf{D}}\int_{\mathbb{D}}\left[\frac{C(\epsilon)+\epsilon f}{x}\right]\mathrm{d}\mathbb{F}
\leq\lim\limits_{x\rightarrow\infty}\sup\limits_{f\in\mathbf{B}(x)-\mathbf{D}}\int_{\mathbb{D}}\left[\frac{C(\epsilon)}{x}+\epsilon\rho(\mathbb{I}_{\mathbb{D}})\right]\mathrm{d}\mathbb{F}=\epsilon\rho(\mathbb{I}_{\mathbb{D}}).
\end{split}
\end{equation}
Consequently, $u'(\infty)=0$.

To show the validity of the three assertions as follows
\begin{equation}\nonumber
\int_{\mathbb{D}}\hat{f}(x)\hat{g}(y)\mathrm{d}\mathbb{F}=h(x)y,~\text{   }~u'(x)=\int_{\mathbb{D}}\frac{h'(x)}{h(x)}\hat{f}\phi'(\hat{f})\mathrm{d}\mathbb{F},~\text{   }~v'(y)=\int_{\mathbb{D}}\frac{\hat{g}\psi'(\hat{g})}{y}\mathrm{d}\mathbb{F}.
\end{equation}
We have established the third one in Lemma \ref{l8}. The other two assertions are simply reformulations, when we use the relations $y=\frac{u'(x)}{h'(x)}$, $h(x)=v'(y)$, $\hat{f}(x)=-\psi'\left(\hat{g}(y)\right)$ and $\hat{g}(y)=\phi'\left(\hat{f}(x)\right)$.

{Thus,} the proof of Theorem \ref{solution} is completed.
\end{proof}
}}

\section{ {\bf{Proof of Theorem \ref{TT5}}. } }\label{proof of TT5}
\begin{proof}[{ \bf{Proof of Theorem \ref{TT5}} }]
Because $\psi$ is strictly decreasing, there exists a weight function $\hat{\lambda}\in\mathbf{L}^0_+(\mathbb{D}, \mathcal{B}(\mathbb{D}), \mathbb{F})$ satisfying  $<\hat{\lambda},\mathbb{I}_{\mathbb{D}}>=1$ and the solution to Problem~(\ref{D2}) has the following form:
\begin{equation}\nonumber
\hat{g}(y)=\frac{y\hat{\lambda}}{\rho(\hat{\lambda})}.
\end{equation}
Based on Lemma \ref{l9}, we know
\begin{equation}\nonumber
\phi'\left(\hat{f}(x)\right)=\frac{y\hat{\lambda}}{\rho(\hat{\lambda})},
\end{equation}
which means
\begin{equation}\nonumber
\hat{\lambda}\propto \phi'\left(\hat{f}(x)\right),
\end{equation}
where $f\propto g$ means that $f$ is proportional to $g$. Using Lemma \ref{l8}, we have
\begin{equation}\nonumber
\begin{split}
h(x)y&=-v'(y)y\\
&=\int_{\mathbb{D}}\hat{g}(y)\left(\phi'\right)^{-1}(\hat{g}(y))\mathrm{d}\mathbb{F}\\
&=\int_{\mathbb{D}}\hat{g}(y)\hat{f}(x)\mathrm{d}\mathbb{F}\\
&=<\hat{f}(\cdot~;x),\frac{y\hat{\lambda}(\cdot)}{\rho(\hat{\lambda})}>,
\end{split}
\end{equation}
i.e., $<\hat{f}(\cdot~;x),\hat{\lambda}(\cdot)>=h(x)\rho(\hat{\lambda})$. As such, $\hat{\lambda}$ is the weight function corresponding to $\hat{f}(x)$ in Theorem \ref{TT1}, and based on {Theorem 2.2} 
in \cite{25}, the optimal terminal wealth ${{\hat{X}_T}}$ is given by
\begin{equation}\nonumber
{{\hat{X}_T}} = I\left(\hat{Y}^{\hat{\lambda}}_T\right){e^{rT}}.
\end{equation}
The rest of the results in Theorem \ref{TT5} are also given by the results of classical EUT problem in \cite{25}.
\end{proof}

\begin{thebibliography}{99}
\baselineskip 16pt
\bibitem[Backhoff and Fontbona(2016)]{Back16}Backhoff Veraguas, J.D., Fontbona, J., 2016. Robust utility maximization without model compactness. SIAM Journal on Financial Mathematics, 7(1), 70-103.
\bibitem[Balter, Mahayni and Schweizer(2021)]{Balter et21} Balter, A. G., Mahayni, A., Schweizer, N., 2021. Time-consistency of optimal investment under smooth ambiguity. European Journal of Operational Research, 293(2), 643-657.
\bibitem[Bartl, Kupper and Neufeld(2021)]{Bartl21} Bartl, D., Kupper, M. and Neufeld, A. Duality theory for robust utility maximisation. Finance Stoch 25, 469–503 (2021).
\bibitem[Biagini and Frittelli(2005)]{Bia05}Biagini, S., Frittelli, M., 2005. Utility maximization in incomplete markets for unbounded processes. Finance and Stochastics, 9(4), 493-517.
\bibitem[Bianchi and Tallon(2019)]{BT18}Bianchi, M., Tallon, J.M., 2019. Ambiguity preferences and portfolio choices: evidence from the field.  Management Science, 65(4), 1486-1501.
\bibitem[Bj\"{o}rk, Khapko and Murgoci(2017)]{BJ17} Bj\"{o}rk, T., Khapko, M., Murgoci, A., 2017. On time-inconsistent stochastic control in continuous time. Finance and Stochastics, 21, pp.~331-360.
\bibitem[Bj\"{o}rk, Murgoci and Zhou(2013)]{BJ13}Bj\"{o}rk, T., Murgoci, A., Zhou, X. Y., 2013.  Mean-variance portfolio optimization with state-dependent risk aversion. Mathematical Finance, 24, 1-24.
\bibitem[Blanchard, Shiller and Siegel(1993)]{Blanchard93}Blanchard, O.J., Shiller, R.,  Siegel, J.J., 1993. Movements in the equity premium. Brookings Papers on Economic Activity, 1993(2), 75-138.
\bibitem[Chen, Ju and Miao(2014)]{Chen14}Chen, H., Ju, N., Miao, J., 2014. Dynamic asset allocation with ambiguous return predictability. Review of Economic Dynamics, 17(4), 799-823.
\bibitem[Collin-Dufresne, Johannes and  Lochstoer(2016)]{Co16}Collin-Dufresne, P., Johannes, M.,  Lochstoer, L.A., 2016. Parameter learning in general equilibrium: The asset pricing implications. American Economic Review, 106(3), 664-98.
\bibitem[Cox and Huang(1989)]{2}Cox, J.C.,  Huang, C.F., 1989. Optimal consumption and portfolio policies when asset prices follow a diffusion process. Journal of Economic Theory, 49, 33-83.
\bibitem[Cvitanic, Schachermayer and Wang(2001)]{C01}Cvitanic, J., Schachermayer, W.,  Wang, H., 2001. Utility maximization in incomplete markets with random endowment. Finance and Stochastics, 5(2), 259-272.
\bibitem[Delbaen and Schachermayer(1995)]{DS95}Delbaen, F. and Schachermayer, W. 1995. The no-arbitrage property under a change of num\'{e}raire. Stochastics and Stochastic Reports, 53, 213–226.
\bibitem[Ellsberg(1961)]{E61}Ellsberg, D., 1961. Risk, ambiguity, and the Savage axioms. Quarterly Journal of Economics, 75, 643-669.
\bibitem[Ekeland, Mbodji and Pirvu(2012)]{Ekeland et12} Ekeland, I., Mbodji, O., Pirvu,  T. A., 2012. Time-consistent portfolio management. SIAM Journal on Financial Mathematics, 3, pp.~1-32.
\bibitem[Gilboa and Schmeidler(1989)]{Gil89}Gilboa, I.,  Schmeidler, D., 1989. Maxmin expected utility with non-unique prior. Journal of Mathematical Economics, 18, 141-153.
\bibitem[Gundel(2005)]{Gun05}Gundel, A., 2005. Robust utility maximization for complete and incomplete market models. Finance and Stochastics, 9(2), 151-176.
\bibitem[Hu, Imkeller and M\"{u}ller(2005)]{Hu05}Hu, Y., Imkeller, P., M\"{u}ller, M., 2005. Utility maximization in incomplete markets. The Annals of Applied Probability, 15(3), 1691-1712.
\bibitem[Hugonnier and Kramkov(2004)]{Hu04}Hugonnier, J.,  Kramkov, D., 2004. Optimal investment with random endowments in incomplete markets. The Annals of Applied Probability, 14(2), 845-864.
\bibitem[Hu, Jin and Zhou(2017)]{Hu et17}  Hu, Y., Jin, H., Zhou, X. Y., 2017. Time-inconsistent stochastic linear-quadratic control: characterization and uniqueness of equilibrium. SIAM Journal on Control and Optimization, 55(2017), 1261-1279.
\bibitem[Li and Zheng(2018)]{Li18}Li, Y., Zheng, H., 2018. Dynamic convex duality in constrained utility maximization. Stochastics, 90(8), 1145-1169.
\bibitem[Lin and Yang(2017)]{Lin17}Lin, Y., Yang, J., 2017. Utility maximization problem with random endowment and transaction costs: when wealth may become negative. Stochastic Analysis and Applications, 35(2), 257-278.
\bibitem[Ju and Miao(2012)]{Ju12}Ju, N.,  Miao, J., 2012. Ambiguity, learning, and asset returns. Econometrica, 80(2), 559-591.
\bibitem[Klibanoff, Marinacci and Mukerji(2005)]{7}Klibanoff, P., Marinacci, M., Mukerji, S., 2005. A smooth model of decision making under ambiguity. Econometrica,  73, 1849-1892.
\bibitem[Klibanoff, Marinacci and Mukerji(2009)]{Klibanoff}Klibanoff, P., Marinacci, M.,  Mukerji, S., 2009. Recursive smooth ambiguity preferences. Journal of Economic Theory, 144, 930-976.
\bibitem[Kramkov and Schachermayer(1999)]{25}Kramkov, D., Schachermayer, W., 1999. The asymptotic elasticity of utility functions and optimal investment in incomplete markets. The Annals of Applied Probability,  904-950.
\bibitem[Neufeld and Marcel(2018)]{Neu}Neufeld, A., Marcel, N., 2018.  Robust utility maximization with L\'{e}vy processes.  Mathematical Finance, 28(1),  82-105.
\bibitem[Owen and Zitkovic(2009)]{Owen09}Owen, M. P., Zitkovic, G., 2009. Optimal investment with an unbounded random endowment and utility-based pricing. Mathematical Finance, 19(1), 129-159.
\bibitem[Schied(2005)]{Sch05}Schied, A., 2005. Optimal investments for robust utility functionals in complete market models. Mathematics of Operations Research, 30(3), 750-764.
\bibitem[Schied(2007)]{Sch07}Schied, A., 2007. Optimal investments for risk-and ambiguity-averse preferences: a duality approach. Finance and Stochastics, 11(1), 107-129.
\bibitem[Schied(2008)]{Sch08}Schied, A., 2008. Robust optimal control for a consumption-investment problem. Mathematical Methods of Operations Research, 67(1), 1-20.
\bibitem[ Strasser(1985)]{maxmin}Strasser, H., 1985. Mathematical Theory of Statistics: Statistical Experiments and Asymptotic. Decision Theory. de Gruyter, Berlin.
\bibitem[Tevzadze, Toronjadze and Uzunashvili(2013)]{Tev13}Tevzadze, R., Toronjadze, T., Uzunashvili, T., 2013. Robust utility maximization for a diffusion market model with misspecified coefficients. Finance and Stochastics, 17(3), 535-563.
\bibitem[Vigna(2020)]{Vigna20}Vigna, E., 2016.  On time consistency for mean-variance portfolio selection. Collegio Carlo Alberto Notebook, 476.
\bibitem[Wittm\"{u}ss(2008)]{Wit} Wittm\"{u}ss, W., 2008. Robust optimization of consumption with random endowment. Stochastics, 80, pp.~459--475.
\bibitem[Zhou and Li(2000)]{Zhou}Zhou, X.Y., Li, D., 2000. Continuous-time mean-variance portfolio selection: a stochastic LQ framework. Applied Mathematics and Optimization, 42, 19-33.


\end{thebibliography}

\end{document}